\documentclass[a4paper,USenglish,numberwithinsect]{lipics-v2021}

\nolinenumbers
\hideLIPIcs

\usepackage{booktabs}
\usepackage{cite}
\usepackage{xspace}
\usepackage{bm}
\usepackage{bold-extra}

\newcommand{\bbN}{\mathbb{N}}
\newcommand{\bbR}{\mathbb{R}}
\DeclareMathOperator{\aggr}{aggr}
\DeclareMathOperator{\asum}{sum}
\DeclareMathOperator{\obj}{obj}

\newcommand{\vmat}[2]{v\bigl(\begin{smallmatrix}#1\\#2\end{smallmatrix}\bigr)}
\newcommand{\highlight}[1]{\textbf{\boldmath #1}}

\newcommand{\local}{\ensuremath{\mathsf{LOCAL}}\xspace}
\newcommand{\pn}{\ensuremath{\mathsf{PN}}\xspace}
\DeclareMathOperator{\ID}{ID}

\renewcommand{\vec}[1]{\mathbf{#1}}

\newcommand{\A}{\text{\textup{A}}\xspace}
\newcommand{\ONL}{\ensuremath{A}}
\newcommand{\OFF}{\OPT}
\newcommand{\OPT}{\text{\textup{OPT}}\xspace}
\newcommand{\ALG}{\text{\textup{ALG}}\xspace}
\newcommand{\mis}{\text{\textup{mis}}\xspace}

\newcommand{\seg}[1]{\textbf{\footnotesize #1}}
\renewcommand{\P}{\seg{P}}
\newcommand{\seq}{\bm{\sigma}}
\renewcommand{\L}{\seg{L}}
\newcommand{\R}{\seg{R}}
\newcommand{\W}{\seg{W}}
\newcommand{\V}{\seg{V}}
\newcommand{\U}{\seg{U}}
\renewcommand{\S}{\seg{S}}
\newcommand{\F}{\seg{F}}

\newcommand{\ac}{\mathsf{a}}
\newcommand{\bc}{\mathsf{b}}
\newcommand{\cc}{\mathsf{c}}

\newcommand{\LRU}{\textup{\textsf{LRU}}\ensuremath{_T}\xspace}
\newcommand{\costmig}{d}
\newcommand{\additive}{\alpha}

\newcommand{\first}{\textrm{first}}
\newcommand{\last}{\textrm{last}}

\newcommand{\myref}[2]{\hyperref[#2]{\ref*{#1}.\ref*{#2}}}

\hypersetup{
    colorlinks=true,
    linkcolor=black,
    citecolor=black,
    filecolor=black,
    urlcolor=black
}

\title{Temporal Locality in Online Algorithms}

\author{Maciej Pacut}{Technical University of Berlin, Germany}{maciej@inet.tu-berlin.de}{https://orcid.org/0000-0002-6379-1490}{}
\author{Mahmoud Parham}{University of Vienna, Austria}{mahmoud.parham@univie.ac.at}{https://orcid.org/0000-0002-6211-077X}{}
\author{Joel Rybicki}{IST Austria, Austria}{joel.rybicki@ist.ac.at}{https://orcid.org/0000-0002-6432-6646}{}
\author{Stefan Schmid}{Technical University of Berlin, Germany and Fraunhofer SIT, Germany}{stefan_schmid@tu-berlin.de}{https://orcid.org/0000-0002-7798-1711}{}
\author{Jukka Suomela}{Aalto University, Finland}{jukka.suomela@aalto.fi}{https://orcid.org/0000-0001-6117-8089}{}
\author{Aleksandr Tereshchenko}{Aalto University, Finland}{aleksandr.tereshchenko@aalto.fi}{}{}

\authorrunning{M.~Pacut, M.~Parham, J.~Rybicki, S.~Schmid, J.~Suomela, and A.~Tereshchenko}
\Copyright{Maciej Pacut, Mahmoud Parham, Joel Rybicki, Stefan Schmid, Jukka Suomela, and Aleksandr Tereshchenko}



\acknowledgements{
  This research has received funding from
	European Union’s Horizon 2020 research and innovation programme,
	under
	the European Research Council (ERC)
  grant agreement No.~864228, 
  the Marie Sk{\l}odowska-Curie grant agreement No.~840605, and from the Austrian Science Fund (FWF) as well as the German Research Foundation (DFG),  project I 4800-N  (ADVISE), 2020-2023.
  }

\ArticleNo{\mbox{}}

\begin{document}

\maketitle

\begin{abstract}
  Online algorithms make decisions based on past inputs. In general, the decision may depend on the entire history of inputs. If many computers run the same online algorithm with the same input stream but start at different times, they do not necessarily make consistent decisions.

  In this work, we introduce \emph{time-local online algorithms}. These are online algorithms where the output at a~given time only depends on $T = O(1)$ most recent inputs. The use of (deterministic) time-local algorithms in a~distributed setting automatically leads to globally consistent decisions.
  
  We revisit caching to explore the competitiveness of classic online problems from the perspective of locality, deriving upper and lower bounds.
  The simplicity of time-local algorithms enable an \emph{algorithm synthesis method} for e.g. metrical task systems, that one can use to design optimal time-local online algorithms for small values of $T$. We demonstrate the power of synthesis in the context of a~variant of the \emph{online file migration problem}.

  We consider a~simple addition of a~clock (counting the number of inputs seen so far) to time-local algorithms, which adds significant power.
  A~large class of online algorithms that have access to \emph{all past inputs} can be transformed to competitive \emph{clocked} time-local algorithms, which implies competitive time-local algorithms for e.g., list access and binary search trees.
\end{abstract}

\newpage

\section{Introduction}
\label{sec:intro}

Online algorithms \cite{Borodin1998} make decisions based on past inputs, with the goal of being competitive against an algorithm that sees also future inputs. On the way towards optimal competitiveness, some algorithms, such as work function algorithms for metrical task systems and the $k$-server problem \cite[Ch. 9 and 10]{Borodin1998}, require access to all past inputs to make decisions, which is storage-expensive.

Some simpler but still highly competitive algorithms deliberately look only a bounded number of inputs into the past. Examples include some algorithms operating in phases: for file migration~\cite{AwerbuchBF03} or binary search trees~\cite[Ch. 1]{FiatW96}. Despite solid presence of algorithms that forget the past in the online literature, such algorithms are still not well-understood.

In this work, we introduce \emph{time-local online algorithms}; these are online algorithms in which the output at any given time is a~function of only $T$ latest inputs (instead of the full history of past inputs). By forgetting the past and having limited access to the input, time-local algorithms gain new attractive properties, which are not exhibited by general online algorithms.
Let us give three motivating examples.

\subparagraph*{Fault-Tolerant Distributed Decision.}
Time-local online algorithms lead to fault-tolerant distributed decision-making.
Consider a~setting in which many geographically distributed computers need to make \emph{consistent} decisions. All computers can observe the same input stream, and each day each of them has to announce its own decision.

If all computers are started at the same time, we can take any deterministic online algorithm and let each computer run its own copy of the algorithm. However, this approach does not \emph{tolerate failures}: if a~computer crashes and is restarted, the local state of the algorithm is lost, and as the decisions may depend in general on the full history of inputs, it will no longer make consistent decisions with the others.

Deterministic time-local online algorithms provide automatically the guarantee that all computers will make consistent decisions. The system will tolerate an arbitrary number of failures and ensure that the computers will also recover from transient faults, i.e., it is \emph{self-stabilizing}~\cite{Dolev2000,Dijkstra1974self-stabilizing}: in $T$ steps since the latest failure, all computers will deterministically make consistent decisions, without any communication.

\subparagraph*{Random Access to the Decision History.}

The second benefit of time-local online algorithms is that they make it possible to efficiently access any past decision with zero additional storage beyond the storage of the input stream. To recover a past decision at any time~$i$, it is sufficient to look up the last~$T$ inputs at time~$i$ and apply the deterministic time-local algorithm. With classic online algorithms, one would have to either store the decision, store the local state, or re-run the entire algorithm up to point $i$.

\subparagraph*{Automated Synthesis of Optimal Time-Local Algorthms.}
The simplicity of time-local algorithms allows to use \emph{computational techniques} to automate the design of time-local algorithms.
We describe and implement a~novel \emph{algorithm synthesis method} that allows us to automate the design of optimal regular time-local algorithms for a~class of \emph{local optimization problems}.

\subsection{Model: Online Problems and Time-Local Algorithms}
\label{ssec:model-online}

\subparagraph*{Online Problems as Request-Answer Games.}
Online problems are often formalized in the request-answer game framework~\cite[Ch. 7]{Borodin1998} introduced by Ben-David et al.~\cite{Ben-David1994}. A~request-answer game consists of an input set $X$ (requests), an output set $Y$ (answers), and an infinite sequence $(f_n)_{n \ge 1}$ of cost functions
\[
    f_n \colon X^n \times Y^n \to \mathbb{R} \cup \{\infty\} \textrm{ for each } n \in \mathbb{N}.
\]
The \emph{optimal offline cost} of an input sequence $ \vec x \in X^n$ is
$\OPT(\vec x) = \min \{ f_n(\vec x, \vec y) : \vec y \in Y^n \}$.

\subparagraph*{Classic Online Algorithms.}

An online algorithm $A$ in the classic sense, i.e., an algorithm that has access to all past inputs, can be defined as
a sequence $(A_i)_{i \ge 1}$ of functions $A_i \colon X^{i-1} \to Y$. The output $\vec y = A(\vec x)$ of the algorithm on input $\vec x \in X^n$ is given by
\[
    y_i = A_i(x_1, \ldots, x_{i-1}) \textrm{ for each } 1 \le i \le n.
\]
The quality of an online algorithm is measured by comparing the cost of its output against the optimal offline cost. An algorithm is said to be $c$-competitive (have a~\emph{competitive ratio} $c$) if for any input sequence $\vec x \in X^n$ its output $\vec y = A(\vec x)$ satisfies $f_n(\vec x, \vec y) \le c \cdot \OPT(\vec x) + \additive$ for a~fixed constant~$\alpha$. We say that an algorithm is \emph{strictly $c$-competitive} if additionally $\additive=0$.

\subparagraph*{Time-Local Online Algorithms.}
Fix $T \in \mathbb{N}$. A~time-local algorithm that has access to $T$ latest inputs is given by a~sequence of maps $(A_i)_{i \ge 1}$ of the form $A_i \colon (X \cup \{\bot\})^{T} \to Y$, where $\bot \notin X$. The output of the algorithm is given by
\[
    y_i = A_i(x_{i-T}, \ldots, x_{i-1}) \textrm{ for each } 1 \le i \le n,
\]
where we let $x_j = \bot$ be placeholder values for $j < 1$.

The request/answer counter $i$ is referred to as a~\emph{clock}.
We say that the algorithm is \emph{regular} if all maps $A_i$ are identical and otherwise it is \emph{clocked}. That is, in the latter case the $i$th decision $y_i$ made by the algorithm may depend on the current time step $i$.

We depict the relations between the classes of algorithms in Figure~\ref{fig:overview}, where we also include more general local distributed algorithms (which we discuss in detail in Section~\ref{sec:temporal-vs-spatial-locality}).

\subsection{Contributions}

In this paper, we formalize the new notion of temporal locality in online algorithms. We investigate the power of time-local online algorithms by focusing on two basic models, regular and clocked algorithms.
We give a~series of results and techniques that illustrate different aspects of time-local algorithms.

\subparagraph*{Charting the Landscape: Time-Local Algorithms for Known Online Problems.}
Do competitive time-local online algorithms even exist for classic online problems?
What are the trade-offs between locality and competitiveness ---
how much does the quality of solutions improve if we allow the algorithms to see further into the past? In particular,
  for a~given $T \ge 1$, what is the best achievable competitive ratio for a~time-local online algorithm that makes decisions based on the previous $T$ inputs? 

We find out that despite their restricted access to input, time-local algorithms can provide competitive solutions for many online problems.
We characterize competitiveness of time-local algorithms for classic online problems such as caching~\cite{Sleator1985} and file migration~\cite{Bartal1995}, and study the tradeoffs between locality and competitiveness for these problems.

\subparagraph*{Synthesis of Time-Local Online Algorithms.}
For online problems including metrical task systems~\cite{Borodin1992}, we automate the synthesis of optimal time-local algorithms.
By leveraging the connection to local graph algorithms, we describe and implement a~novel \emph{algorithm synthesis method} that allows us to automate the design of optimal regular time-local algorithms for a~class of \emph{local optimization problems}.
Specifically, the synthesis task can be formulated as a~certain weighted optimization problem in dual de Bruijn graphs.

For our case study problem of online file migration, we synthesize optimal deterministic algorithms for small values of $T$ and a~large range of $\costmig$. For example, we show that for unit costs ($\costmig=1$) there exists a~$3$-competitive time-local algorithm with $T=4$, which is the best competitive ratio achieved by \emph{any} deterministic online algorithm~\cite{Black1989}. Moreover, we describe how to extend our synthesis framework to obtain efficient 
\textit{randomized} algorithms.

\subparagraph*{The Power of Knowing the Time.}
We will see that some problems do not admit competitive \emph{regular} time-local algorithms. Motivated by this, we also investigate the power of \emph{clocked} algorithms, i.e., algorithms that know how many inputs have been processed so far. How much does this additional information help in obtaining competitive algorithms for problems that do not admit regular time-local algorithms?

We demonstrate that \emph{clocked} time-local algorithms can be powerful: for a~large class of online problems, classic (full-history) online algorithms can be automatically translated into \emph{clocked} time-local algorithms with negligible overhead to the competitive~ratio. 
This implies competitive clocked time-local algorithms for many online problems, including e.g., online list access~\cite{Sleator1985} and binary search trees~\cite{demaine2007dyynamic}.
This generalizes a known result for binary search trees to bounded monotone games: any online algorithm for binary search trees can be forced into a canonical state every $c\cdot n$ operations without affecting the runtime by more than $2
\cdot(2n-6)/c\cdot n$ times the competitive ratio per operation ($2n-6$ is the maximum rotation distance between trees ~\cite[Ch. 1]{FiatW96}). This is why the binary search trees literature only focuses on sequences of length $n$, and with the Theorem~\ref{thm:simulation}, we can consider fixed size sequences, with the size depending on the game delay and diameter.
Further, we find negative results: for some problems clock does not help to achieve high competitiveness.

\subparagraph*{Temporal vs.~Spatial Locality: Online Algorithms and Distributed Computing Meet.}
  We explore the connections between different models studied in \emph{distributed graph algorithms} and different variants of \emph{time-local online algorithms}. 
Distributed algorithms make decisions based on the \emph{local information in the spatial dimension}, while time-local online algorithms make decisions based on the \emph{local information in the temporal dimension}; see Figure~\ref{fig:overview}.
We exploit this connection, and discuss how to lift some results from theory of distributed computing to establish impossibility results for time-local algorithms.

\begin{figure}
\centering
\includegraphics[page=1,scale=0.8]{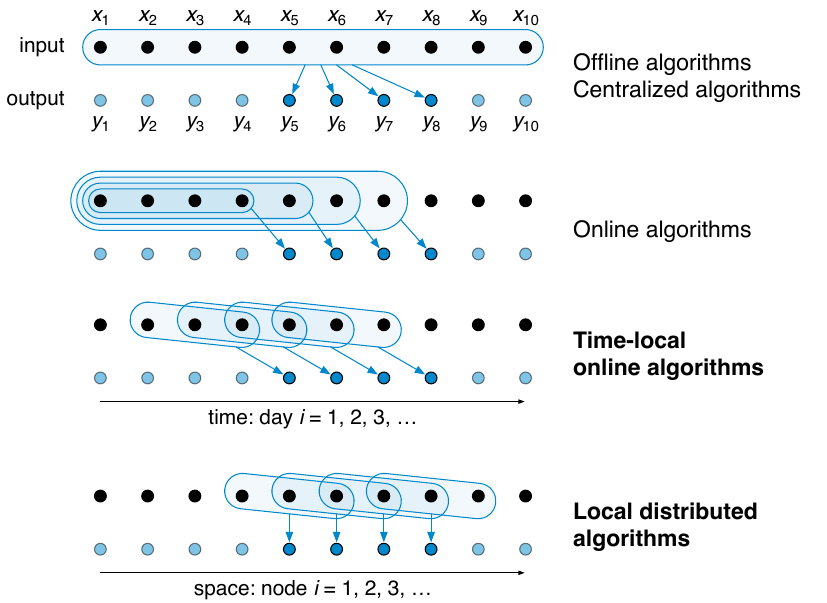}
\caption{Local decision-making in time vs.\ space dimensions.}\label{fig:overview}
\end{figure}

\subsection{Prior Work on Restricted Models of Online Computation}

  To our best knowledge, \emph{temporal locality} of online algorithms has not been systematically studied.
  However, other restricted forms of online algorithms have received some attention.
  For example, Chrobak and Larmore~\cite{Chrobak1991} introduced the notion of \emph{memoryless} online algorithms, defined for online problems with explicit notion of an external configuration of the algorithm, and the costs of transitioning between the configurations (captured by metrical task systems~\cite{Borodin1992}).
  In memoryless algorithms, the answer to the current request can only depend on the current configuration instead of being an arbitrary function of the entire past history as in general online algorithms.
  However, memoryless online algorithms differ from time-local algorithms, as memoryless algorithms have access to the configuration of the algorithm, whereas time-local algorithms are unaware of the configuration, outputted a moment ago while responding to the previous input.
  The lack of access to the current configuration distinguishes time-local algorithms from memoryless algorithms~\cite{Chrobak1991}: memoryless algorithms can store information about the past in the online algorithm's configuration, whereas time-local algorithms make decisions on the $T$ latest inputs.
  Similarly to time-local algorithms, memoryless online algorithms can be synthesized using a~fixed point approach~\cite{Chrobak1991}. 

  Ben-David et al.~\cite{Ben-David1994} investigated local online problems  within the request-answer game framework of online algorithms. 

  However, their notion of locality applies to the \emph{cost functions} defining the online problem instead of the \emph{algorithms} solving them: for these problems,
  the cost of a~solution cannot depend on inputs too far in the past.
In this work, to avoid confusion, we later refer to these games as \emph{bounded delay games}.

\section{Regular Time-Local Algorithms}
\label{sec:regular-algorithms}

Deterministic  regular time-local algorithms are functions of the last $T$ inputs.
A~time-local algorithm is given by a~map $A \colon (X \cup \{\bot\})^{T} \to Y$, where $T \in \mathbb{N}$ and $\bot \notin X$.
At time $i$, the output $y_i$ of the algorithm is given by
\[
    y_i = A(x_{i-T}, \ldots, x_{i-1}) \textrm{ for each } 1 \le i \le n,
\]
where we let $x_j = \bot$ be placeholder values for $j < 1$.

In this section, we explore competitiveness of deterministic regular online time-local algorithms given by the above definition.
We revisit the classic problem of caching with 
the goal to design time-local algorithms performing as closely as possible to an optimal offline algorithm.
Then, we introduce a novel synthesis method, which we use to synthesize optimal algorithms for online file migration problem. We generalize and analyze the synthesized algorithms and derive and lower bounds for the problem.

\subsection{Time-Local Algorithms for Caching}
\label{ssec:caching}

\subparagraph*{Online Caching Problem.}
In the \emph{online caching} problem~\cite{Sleator1985}, we manage a two-level memory hierarchy, consisting of a slow memory, storing the set of all $n$ pages, and a fast memory, called \emph{cache} that can store any size $k$ subset of pages.
We are given a sequence $\sigma$ of requests to the pages.
If a requested page is not in the cache, a \emph{page fault} occurs, and the page must be moved to the cache. As the size of the cache is limited, we must specify which page to evict to make space for the requested page.
The goal is to minimize the number of page faults.

Note that in general, 
there is no unique way to encode an online problem as a request-answer game.
When considering the time-local setting, the encoding of a problem
should avoid complex actions with e.g.\ effects depending on the past. In our encoding of the caching problem, the input set $X = \{ 1, \ldots, n \}$ coincides with the set of all pages, and the set
$Y = \{ A \subseteq X : |A| \le k \}$ of outputs coincides with possible cache configurations (sets of at most $k$ pages).

\subparagraph*{A Lower Bound.}
We start with a~simple lower bound for deterministic regular time-local algorithms,
showing that in the worst-case, the cost of any time-local algorithm for the caching problem inevitably grows with the input sequence.

\begin{theorem}
  Every deterministic time-local online algorithm with a fixed visible horizon $T$ for caching incurs the cost at least $\min\{k-1, n-k\}\cdot |\sigma| / (T+k-1)$, where $k$ is the size of the cache and $n$ is the number of pages, and $\sigma$ is the input sequence.
  This cost can be incurred even if the cost of an optimal offline algorithm is bounded by a constant.
  \label{thm:additive-caching}
\end{theorem}

\begin{proof}
  Deterministic time-local algorithms are simply functions of the last $T$ requests, hence the algorithm's output is fixed when its visible horizon is $\ac^T$, for any page $\ac$. Fix any deterministic time-local online algorithm and a page $\ac$, and let us denote the output of the algorithm on~$\ac^T$ as the~\emph{default} configuration $D$.
  We select a set $B$ of at most $k-1$ pages not present in $D$.
  At least $n-k$ pages are not present in $D$, hence the size of $B$ is $\min\{ k-1, n-k\}$.
  Let $\vec{b} = \langle b_1, b_2, b_3, \ldots, b_{\min\{ k-1, n-k\}} \rangle$ be a sequence of pages from $B$ ordered in an arbitrary fashion.

If $D$ does not contain $\ac$, then for any input $\sigma$ consisting of requests to $\ac$ only the algorithm incurs cost $|\sigma|$, and the claim follows.
Otherwise, consider an input sequence $\sigma := (\ac^T \cdot \vec{b})^L$ for some $L\in\mathbb{N}$.
We partition $\sigma$ into $L$ \emph{phases} of form $\ac^T \cdot \vec{b}$.
Fix any phase. 
After serving the subsequence $\ac^T$, the online algorithm resides in the default configuration $D$.
Since $D$ does not contain any page from $B$, the subsequent $|B|$ requests incur the cost $1$ each.
Hence, in each phase the online algorithm pays at least $|B| = \min\{k-1,n-k\}$.
Summing over all $L = |\sigma|/(T+k-1)$ phases, the total cost of the algorithm is at least $\min\{k-1,n-k\}\cdot |\sigma| / (T+k-1)$.

A feasible offline solution for $\sigma$ is to move to the configuration $\{\ac\} \cup B$ at the beginning, incurring the cost at most $k$ for reaching it. In this configuration, all requests from $\sigma$ are free, hence the cost of an optimal offline solution is at most $k$.
\end{proof}

Consequently, no deterministic time-local algorithm for caching can be competitive in the classic sense. 
Despite this negative result, can we still design time-local algorithms that have performance close to the offline optimum?
Observe that we may lower the cost from Theorem~\ref{thm:additive-caching} by increasing $T$. The actual value of $T$ is often under control of the system designer, who may supply more storage to diminish the cost incurred by the online algorithm. Next, we study how closely time-local algorithm can perform to an offline optimum for a~given~$T$.

\subparagraph*{Competitive Ratio with a Periodic Additive.}
To characterize caching in the time-local setting, we extend the notion of competitiveness to incorporate an additive cost that grows with the input length.
We say an algorithm is \emph{$c_1$-competitive with a periodic additive cost of $c_2$} if there exists a~constant $c_3$ such that for each input $\sigma$ we have
\[
\ALG(\sigma)\le c_1 \cdot \OPT + c_2\cdot |\sigma| / T + c_3.
\]

\subparagraph*{Time-Local Variant of Least Recently Used.}
Next, we introduce \LRU, a~natural time-local algorithm for caching, which simulates the classic Least Recently Used~\cite{Sleator1985} algorithm for the last $T$ requests and outputs its configuration (the pages leaving the visible horizon are evicted from the cache). We now analyze its performance in terms of competitiveness with a periodic additive.

\begin{theorem}
  For any input sequence $\sigma$, the algorithm $\LRU$ incurs the cost at most $k\cdot \OPT(\sigma) +  k\cdot |\sigma| / T + k$, where $k$ is the cache size.
  \label{thm:lrut}
\end{theorem}

\begin{proof}
We consider the $k$-phase partition of the input sequence $\sigma$, following the notation from Borodin and El-Yaniv~\cite{Borodin1998}: phase 0 is the empty sequence, and every phase $i>0$ is the maximal sequence following the phase $i-1$ that contains at most $k$ distinct page requests since the start of the $i$th phase.
The analysis of the cost of an offline optimal algorithm \OPT repeats the arguments from the classic analysis of the LRU~\cite{Sleator1985,Borodin1998}: \OPT pays at least $1$ in each phase $i>0$ but the last one.
Next, we bound the cost of \LRU.

Consider any phase $\sigma_p$.
Pages leaving the visible horizon are evicted from the cache, hence unlike the full-history LRU, for the \LRU some of these pages may incur the request cost multiple times in a phase.
Fix a~page $x$ requested in the phase, and consider two consecutive requests to $x$, at $t_1$ and $t_2$.
We claim that if \LRU incurs a page fault at $t_2$, and the requests to $x$ are closer than~$T$ apart, then $k+1$ different pages were requested since $t_1$.
After serving the request at $t_1$, $x$ is the most recently used, and \LRU has $k$ distinct pages in the cache.
For $x$ to leave the cache, \LRU must incur a page fault while $x$ is the least recently used page.
However, if this is the case, at least $k+1$ different pages were requested between $t_1$ and $t_2$: the~$k$ pages including $x$ after serving the request at $t_1$, and the page that swapped $x$ out.

Consequently, requests to $x$ that are closer than~$T$ requests apart in $\sigma$, but are contained within a single phase, do not cause a~page fault.
Throughout the phase, \LRU incurs the cost at most $\lceil m_p/T \rceil$ for page faults of the page $x$, where $m_p$ is the length of the phase $\sigma_p$.
The number of pages that can cause page faults in this phase is at most~$k$, thus in total \LRU incurs the cost at most $k\cdot \lceil m_p / T \rceil$.
Comparing to the cost of \OPT in each phase~$\sigma_p$, we have \[
  \LRU(\sigma_p) \le k\cdot \lceil m_p / T \rceil \le k\cdot(1+m_p/T) \le k \cdot \OPT(\sigma_p) + k\cdot m_p/T,
  \]
  where the last inequality holds for all phases but the last one.
We sum these bounds over all phases; for the last phase we use $\LRU(\sigma_p) \le  k\cdot(1+m_p/T)$. As the sum of the phase lengths is $|\sigma|$, we conclude the proof.
\end{proof}

The used our notion of competitiveness with a periodic additive is justified: it is impossible to bound the absolute cost of \LRU in terms of $k\cdot |\sigma|/T$. If the length of each phase (as defined in the above proof) is $k$, \LRU incurs the absolute cost $|\sigma|$. However, such an input sequence is also costly for an optimal offline algorithm. Comparing the cost of \LRU with the cost of an optimal offline algorithm (with the notion of competitive ratio with periodic additive) allows to bound the portion of the cost growing with the input sequence as a~function of $T$.

Despite the negative result from Theorem~\ref{thm:additive-caching}, with large enough $T$, time-local algorithms for caching may still perform close to an optimal offline algorithm.
In Theorem~\ref{thm:lrut}, we established that \LRU is \emph{competitive with a~periodic additive}, and next we put this definition in context of known measures of quality of online algorithms.
In particular, if an algorithm is competitive with periodic additive, it is also \emph{loosely competitive}~\cite{Young02} (a definition coined by Young to characterize caching).
An online algorithm is $\epsilon$-loosely $c$-competitive if for a substantial fraction of inputs, the algorithm is either $c$-competitive or it incurs a~small absolute cost: $\ALG(\sigma) \leq \max\{c\cdot \OPT(\sigma), \epsilon \cdot |\sigma| \}$.
If an online algorithm is $c_1$-competitive with a~periodic additive $c_2\cdot |\sigma| / T$ and $c_3=0$, then it is also $(2\cdot c_2)$-loosely $(2\cdot c_1)$-competitive for all sequences $\sigma$:
\[
  \ALG(\sigma) \le c_1 \cdot \OPT(\sigma) + \frac{c_2\cdot |\sigma|}{T}  = \frac{2 c_1 \cdot \OPT(\sigma) + \frac{2 c_2\cdot |\sigma|}{T}}{2} \le \max \{ 2 c_1 \cdot \OPT(\sigma), \frac{2 c_2}{T} \cdot |\sigma|\},
\]
where the last step follows by the relation of the average and the maximum.

The relation between the relaxed definitions of competitiveness suggests future directions of research.
Considering \LRU with resource augmentation may lead to improved loose competitiveness of time-local algorithm for caching with techniques introduced by Young~\cite{Young02}, but we leave these studies to future work.

The caching problem is a difficult problem for time-local algorithms. 
Later in this paper we will see that problems such as online file migration~\cite{bienkowski2012migrating}, online list access~\cite{Sleator1985} or binary search trees~\cite{demaine2007dyynamic} admit competitive time-local algorithms in the traditional sense.

\subsection{Synthesis of Time-Local Algorithms for Metrical Task Systems}
\label{ssec:synthesis-main-body}

\begin{figure}
\centering
\includegraphics[width=0.96\textwidth]{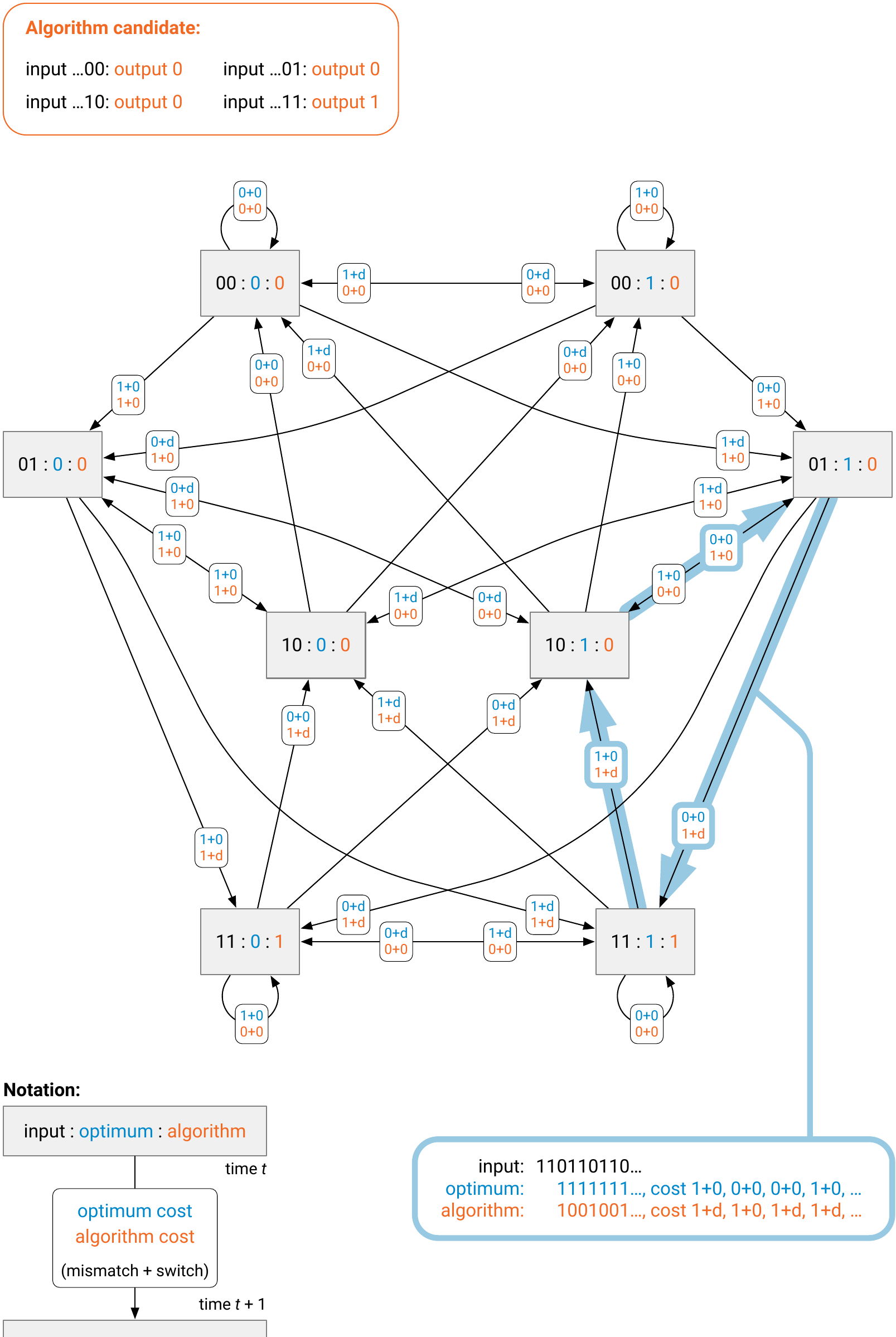}
\caption{Dual de Bruijn graph for online file migration in a 2-node network, with cost of migration~$\costmig$, with time horizon $T = 2$. The highlighted cycle shows how an adversary can force the candidate algorithm to pay $3+2\costmig$ when optimum pays only $1$; hence this specific time-local algorithm cannot be better than $(3+2\costmig)$-competitive.}
\label{fig:synthesis-graph}
\end{figure}

We show how to use \emph{computational techniques} to automate the design of time-local algorithms, by synthesizing optimal time-local algorithms.
This technique allows us to automatically obtain \emph{tight} upper and lower bounds for time-local online algorithms for any given $T$. Our synthesis method applies to a~class of \emph{local optimization problems} (defined formally in Section~\ref{sec:local-problems}), including distributed local problems on paths, but for simplicity of presentation, we consider a less general setting in this section.

In context of online algorithms, our synthesis method applies to e.g.\ \emph{metrical task systems}~\cite{Borodin1992}, and some of its generalizations.
A metrical task system is characterized by the set of states and (metric) costs of transitions between them. The cost of serving a~request depends on the current state of the algorithm. The algorithm may change state before serving each request. The objective is to minimize the total cost.
Moreover, local optimization problems model metrical tasks system variants, where the algorithm is allowed to change the state only after serving the request (as in e.g.\ list access and file migration problems).

In Section~\ref{sec:synthesis},
we show how to construct a~(finite) weighted, directed graph $G(\Pi, A)$ that captures the costs of output sequences as walks in $G(\Pi,A)$: We prove that the for a~large class of local optimization problems, the competitive ratio of any regular time-local algorithm $A$ for $\Pi$ corresponds to the \emph{heaviest directed cycle} in this graph.
\begin{theorem}[informal; see Theorem~\ref{thm:synthesis}]
  Let $\Pi$ be a~local optimization problem and let $A$ be a~regular time-local algorithm with horizon $T$. Then there is a~finite, dual-weighted graph $G=G(\Pi,A)$ such that the competitive ratio of $A$ is determined by the cycle with the heaviest weight ratio in $G$.
\end{theorem}

Recall that each regular time-local algorithm with horizon $T$ is given by some map $A \colon X^T \to Y$. For local optimization problems with finite input set $X$ and output set $Y$, we can iterate through all the $|Y|^{|X|^T}$ maps to find an optimal algorithm for any given $T$.

\subparagraph*{Synthesis Case Study: Online File Migration}
\label{ssec:file-migration-main-body}

We illustrate the usefulness of our synthesis technique by synthesizing several optimal deterministic time-local algorithms for \emph{online file migration} problem~\cite{bienkowski2012migrating}.
Figure~\ref{fig:synthesis-graph} provides an illustration of how the synthesis proceeds in the specific case of the online file migration problem in a $2$-node network, with time horizon $T = 2$. In this case there are $2^T = 4$ possible input sequences that the algorithm may see within its $T$-element window, and for each input sequence the algorithm outputs either $0$ or $1$; and hence there are $2^4 = 16$ possible algorithms. In the figure, we have fixed one possible algorithm candidate $A$.

Then we construct the graph $G(\Pi,A)$, where each node is labeled with a triple ``input : optimum : algorithm.'' For example, the node $11 : 0 : 1$ represents the case that we have input $\ldots 11$, and in this case the optimal output might be $0$, but algorithm $A$ outputs $1$. The number of nodes in this case is $4\cdot 2 = 8$, as there are $4$ possible inputs, and two possible outputs that the optimum might use, but only one output that this specific algorithm will pick.
Then we add edges that represent all possible transitions when we consider how the input may evolve, how the optimum might change, and how the algorithm will respond. For example, we have got an edge from $11 : 0 : 1$ to $10 : 0 : 0$ to indicate that after seeing the input $\ldots 11$, our next input element might be $0$ and hence we have in the next time slot the input sequence $\ldots 10$; in the next time slot the optimum might still keep the value $0$, but our algorithm will switch to output $1$. In this case the optimum will pay $0+0$ units, as there is no mismatch and no need to migrate the file, while the algorithm will pay $1+d$ units, as we needed to serve a file that is at the other node and also we migrated the file to the new node. Therefore, the edge from $11 : 0 : 1$ to $10 : 0 : 0$ is labeled with the weight pair $(0+0, 1+d)$.

Now consider the cycle highlighted in Figure~\ref{fig:synthesis-graph}. This cycle represents a possible input sequence and a possible behavior of the optimum such that the optimum pays only $1+0+0+0+0+0 = 1$ unit and the algorithm pays $1+d+1+0+1+d = 3+2d$ units. Hence, our adversary can generate an unbounded sequence of inputs that follows this cycle, and the algorithm will pay \emph{at least} $3+2d$ times the cost of the offline optimum; hence this algorithm cannot be better than $(3+2d)$-competitive.

Inspired by the synthesized algorithm, we identify an algorithm candidate for general cost of migration $d$, and analytically analyze its competitiveness.
The synthesis helped in this process: for small values of $d$, we generalize them to arbitrary $d$, analyze their competitiveness, and derive asymptotically tight upper and lower bounds.
All technical details and further results are provided in the later sections.
\begin{theorem}[informal; see Theorem~\ref{thm:file-migration-lb}]
  For any $\costmig \ge T$, no randomized regular time-local algorithm achieves a~competitive ratio better than $2\costmig/T$.
\end{theorem}

\begin{theorem}[informal; see Corollary \ref{cor:6-competitive}]
There is a $6$-competitive algorithm for $T \ge 6\costmig$ for any $\costmig \ge 1$. Moreover, for $1 \le T < 6\costmig$ the algorithm is $(4+12\costmig/T)$-competitive.
\end{theorem}

\section{Clocked Time-Local Algorithms}\label{sec:clocked-algorithms}

In this section, we examine the power and limitations of clocked time-local algorithms (defined in Section~\ref{ssec:model-online}).
Adding a clock can simplify the design process of time-local algorithms due to increased expressiveness in comparison to regular time-local algorithms, as demonstrated with the simple example below.

\subparagraph*{Example: Clocked Algorithm Move-To-Min for File Migration.}
Consider the algorithm Move-To-Min~\cite{AwerbuchBF03} for online file migration, which is $7$-competitive for arbitrary networks.
The algorithm Move-To-Min operates in phases of length $d$, and at the end of each phase it moves the file to the node that minimizes the cost of serving all requests from this phase.
The availability of a clock enables us to mimic this strategy: we use the clock to determine the start and the end points of the penultimate phase in comparison to the current input index, and with $T \ge 2d$, the requests from the penultimate phase are still in the visible history.

\medskip 

The remainder of this section is organized as follows.
First, we show that clocked time-local algorithms can be powerful.
For problems that are \emph{bounded monotone} (precisely defined below),
competitive \emph{classic online} algorithms that have access to the full input history can be converted into competitive \emph{clocked time-local} algorithms. Second, we look into implications of this theorem for classic online problems.
Finally, we derive lower bounds for clocked algorithms to demonstrate that the presence of a clock cannot improve competitiveness of time-local algorithms for all online problems.

\subsection{Clocked Time-Local Algorithms from Full-History Algorithms}
\label{ssec:clocked-from-full-history}

We now show that for a~large class of online problems the following result holds: if the problem admits a~\emph{deterministic} classic online algorithm with competitive ratio $c$, then for any given constant $\varepsilon > 0$, there exists a~deterministic clocked time-local algorithm with a~competitive ratio of at most $(1+\varepsilon)c$ for some constant horizon $T$.

The proof follows a~similar structure as the constructive derandomization proof of Ben-David et al.~\cite[Section 4]{Ben-David1994} for classic online algorithms: we chop the input sequence into short segments and show that under certain assumptions, both the offline and competitive online algorithms pay roughly the same cost. However, some care is needed to adapt the proof strategy, as in the case of time-local algorithms, we can only use constant-size segments.

We now define the class of request-answer games for which we prove our result.
A~(minimization) game is
\emph{monotone} if for all~$n \in \mathbb{N}$
\[
f_{n+1}(\vec x \cdot x, \vec y \cdot y) \ge f_n(\vec x, \vec y) \textrm{ for all } \vec x \in X^n, \vec y \in Y^n, x \in X, y \in Y.
\]
That is, the cost cannot decrease when extending the input-output sequence.
We say that a~monotone game has \emph{bounded delay}
if for every $h \in \mathbb{R}$ the set
\[
L(h) = \Bigl\{ \vec x \in \bigcup_{n=1} X^n : \OPT(\vec x) \le h \Bigr\}
\]
is finite (sometimes this property is called locality~\cite{Ben-David1994}). That is, there cannot be arbitrarily long sequences of a~fixed cost: eventually the cost of any sequence must increase.
Finally, the diameter of the game is
\[
D = \sup \Bigl\{ \bigl| f(\vec x \cdot \vec x', \vec y \cdot \vec y') - f(\vec x, \vec y) - f(\vec x', \vec y') \bigr|
: (\vec x, \vec y), (\vec x', \vec y') \in \bigcup_{n > 0} X^n \times Y^n \Bigr \}.
\]
We define that a~\emph{bounded monotone minimization game} is a~monotone minimization game that has bounded delay, finite diameter, and finite input set $X$.
The following result holds for deterministic algorithms:

\begin{theorem}
  \label{thm:simulation}
  Let $\mathcal{F}$ be a~bounded monotone minimization game. If there exists an online algorithm $A$ with competitive ratio $c \geq 1$ for $\mathcal{F}$, then for any constant $\varepsilon > 0$ there exists some constant $T$ and a~clocked $T$-time-local algorithm $B$ with competitive ratio $(1+\varepsilon)c$ for $\mathcal{F}$.
\end{theorem}
\begin{proof}
  Since $A$ has competitive ratio $c$, then there exists some constant $d$ such that for every input $\vec x$ the output $\vec y = A(\vec x)$ satisfies $f(\vec x, \vec y) \le c \cdot \OPT(\vec x) + \additive$.
  Let $D$ be the diameter of the game and fix
$\delta = 2\varepsilon/3$ and $H = (2+\delta)/\delta \cdot \max \{ d, D \}$.
  Since the game has bounded delay, we have that
  $L(H) = \{ \vec x : \OPT(\vec x) \le H \}$ and $T = \max \{ k + 1: (x_1, \ldots, x_k) \in L(H) \}$
  are finite.
  Note that $T$ is independent of $n$, as it only depends on $H$.
  Observe that since the cost functions are monotone, for all $n \ge T$ any input sequence $\vec x \in X^n$ satisfies $\OPT(\vec x) \ge H$.

  We can now construct the \emph{clocked} time-local algorithm that only sees the $T$ latest inputs and the total number of requests served so far. Let $A = (A_i)_{i \ge 1}$ be the classic online algorithm.
  The clocked time-local algorithm $B$ is given by sequence $(B_j)_{j \ge 1}$, where
  \[
B_{Tk+i}(x_{T(k-1)+i}, \ldots, x_{Tk} ,z_1 \ldots, z_{i}) = A_i(z_1, \ldots, z_i) \textrm{ for } 1 \le i \le T \textrm{ and } k \ge 0.
\]
That is, the clocked time local algorithm $B$ simulates the classic online algorithm $A$ and
 resets it every time $T$ inputs have been served since the last reset.

We now analyze the clocked time local algorithm $B$.
For any $n \in \mathbb{N}$, let $\vec x \in X^n$ be some input sequence and $\vec y \in Y^n$ be the output of $B$ on the input sequence~$\vec x$. Let $\vec x(1), \ldots, \vec x(k)$ be the subsequences of $\vec x$, where $\vec x(1)$ denote the first $T$ inputs, $\vec x(2)$, denote the next $T$ inputs, and so on.
Define the shorthand $C(i) = \OPT( \vec x(i))$ for each $1 \le i \le k$. Note that $C(i) \ge H$ for each $1 \le i < k$. The last subsequence $\vec x(k)$ may consist of fewer than $T$ inputs, so we have no lower bound for $C(k)$. For $1 \le i < k$, we get that
\[
C(i) - D \ge \left(\frac{2}{2+\delta}\right) \cdot C(i) \quad \textrm{ and } \quad
D + \additive \le \left(\frac{2\delta}{2+\delta}\right) \cdot C(i)
\]
by applying the fact that $C(i) \ge H$ and the definition of $H$.

By repeatedly applying the definition of diameter, we get that
the optimum offline solution is lower bounded by
\begin{align*}
  \OPT(\vec x) &\ge C(1) + \sum_{i=2}^k \left( C(i) - D \right)
  \ge \sum_{i=1}^{k}( C(i) - D)
  \ge \left( \frac{2}{2+\delta} \right) \sum_{i=1}^k C(i).
\end{align*}
Since $A$ has competitive ratio $c$, the output of $B$ has cost
\begin{align*}
  f(\vec x, B(\vec x)) &\le c \cdot C(1) + \additive + \sum_{i=2}^k \left( c \cdot (C(i) + D) + \additive \right)
  \le \left( c + \frac{2c\delta}{2+\delta} \right) \sum_{i=1}^{k}  C(i) + \additive.
\end{align*}
Now using the lower bound on $\OPT(\vec x)$ and the definition of $\delta$, we get that the output of $B$ has cost bounded by
\begin{align*}
	f(\vec x, B(\vec x)) &\le \left( c + \frac{2c\delta}{2+\delta} \right) \sum_{i=1}^{k}  C(i) + \additive 
	\le \left( c + \frac{2c\delta}{2+\delta} \right) \cdot \left( \frac{2+\delta}{2} \right) \OPT(\vec x) + \additive \\
	&= c \cdot \left(1+ \frac{3\delta}{2} \right) \OPT(\vec x) + \additive =
	c \cdot (1+\varepsilon) \OPT(\vec x) + \additive. \qedhere
\end{align*}
\end{proof}

With the Theorem~\ref{thm:simulation} we have near-optimally competitive time-local online problems for bounded monotone online problems.
For example, we get a $(2+\varepsilon)$-competitive clocked time-local algorithms for the list access problem~\cite{Sleator1985} from Move-to-Front, and $O(\log \log n)$-competitive clocked time-local algorithm for the binary search tree problem~\cite[Ch. 1]{FiatW96}) from the tango trees~\cite{demaine2007dyynamic}.
More broadly, for class of bounded monotone requests answer games includes all metrical task systems~\cite{Borodin1992} that have a property that serving a~request from any configuration always incurs a~positive cost.

\subsection{The Limitations of Clocked Time-Local Algorithms}
\label{ssec:clocked-limitations}

Although access to a clock is a~powerful asset, it does not remedy all problems of time-local computation.
For caching, clocked time-local algorithms cannot be competitive in the classic sense. The additive term defined in Section~\ref{ssec:caching} increasing with the length of the input sequence is unavoidable, even for clocked algorithms.

\begin{theorem}
  The cost of any deterministic clocked time-local algorithm for online caching with cache size $k = 2$ cannot be bounded by $c_1 \cdot \OPT + c_2$, for any constants $c_1, c_2$. 
\end{theorem}

\begin{proof}
  Consider the caching problem with the cache size $k=2$ and a~universe of $n=3$ pages $X = \{\ac, \bc, \cc\}$.
  Let $A$ be any clocked time-local algorithm with horizon $T = O(1)$.
  We say that $A$ is \emph{decisive} if on the infinite sequence $\ac^*$ exists $t$ such that $y_{t'} = y_t$ for all $t' > t$.
  Otherwise, $A$ is \emph{indecisive}; note that any indecisive deterministic algorithm $A$ must be a~clocked algorithm.
  We claim that in both cases (decisive or indecisive), there exists an input sequence for which the online algorithm incurs unbounded cost while an offline algorithm's cost is bounded.

  If $A$ is decisive, we apply the construction from Theorem~\ref{thm:additive-caching}: a~decisive algorithm behaves as a~regular time-local algorithm after $t'$, thus it incurs a~cost strictly growing with the length of the input sequence, while an offline optimum cost is bounded.

  Thus, suppose that $A$ is indecisive and consider an input sequence family $\mathcal{I} := \{ \ac^L : L \in \mathbb{N} \}$.
  An indecisive algorithm changes its output infinitely many times on the infinite sequence consisting only of requests to $\ac$, incurring a~cost growing with the length of the input sequence.
  An optimal offline solution on any input $\vec x \in \mathcal{I}$ incurs constant cost: it may move to the configuration with $\ac$ in the cache at the beginning of the input sequence, and serve the entire sequence without further cost.
 \end{proof}

\section{Local Optimization Problems}	\label{sec:local-problems}

  We define a broad class of distributed and online problems, called local optimization problems. For these problems, the optimal distributed algorithms (wrt. approximation) and the optimal strictly competitive time-local online algorithms are obtainable automatically, see Section~\ref{sec:synthesis}.  
  Moreover, we use local optimization problems to derive a link between distributed computing and local optimization problems (Section~\ref{sec:temporal-vs-spatial-locality}).

  The definition of the class is somewhat technical, but the basic idea is simple:
  at each time step~$i$, the cost (or utility) of our decision $y_i$ is defined to be some function of the current input $x_i$ and up to $r = O(1)$ previous inputs and outputs.
  We apply the formal definition in Section \ref{sec:synthesis} for an algorithmic synthesis of upper and lower bounds.

  This formalism has several attractive features. First, it is flexible enough to define e.g.\ online problems in which we reward correct decisions (e.g.\ whenever we predict correctly $y_i = x_i$, we get some profit), we penalize costly moves (e.g.\ whenever we change our mind and switch to a new output $y_i \ne y_{i-1}$, we get some penalty), and we prevent invalid choices (e.g.\ by defining infinite penalties for decisions that are not compatible with the previous inputs and/or previous decisions). Second, this formalism can  capture problems that are relevant in distributed graph algorithms (e.g.\ $x_i$ represents the weight of node $i$ along a path, $y_i$ indicates which nodes are selected, and we pay $x_i$ whenever we select a node). Finally, this family of problems is amenable to automated algorithm synthesis, as we will later see.

	We will now present the formal definition and then give several examples of different kinds of problems, both from the areas of online and distributed graph algorithms.

	\subparagraph{Formalism.} A \emph{local optimization problem} is a tuple $\Pi = (X,Y,r,v,\allowbreak \aggr,\obj)$, where
	\begin{itemize}
		\item $X$ is the \emph{set of inputs},
		\item $Y$ is the \emph{set of outputs},
		\item $r \in \bbN$ is the \emph{horizon},
		\item $v\colon X^{r+1} \times Y^{r+1} \to \bbR \cup \{-\infty, +\infty\}$ is the \emph{local cost function},
		\item $\aggr \in \{ \asum, \min , \max \} $ is the \emph{aggregation function},
		\item $\obj \in \{ \min, \max \} $ is the \emph{objective}.
	\end{itemize}
	The input for the problem $\Pi$ is a sequence $\vec{x} = (x_1, x_2,\allowbreak \dotsc, x_n) \in X^n$ and a solution is a sequence $\vec{y} = (y_1, y_2,\allowbreak \dotsc, y_n) \in Y^n$.
  For convenience, we will use placeholder values $x_i = y_i = \bot$ for $i < 1$ and $i > n$. With each index, we associate a \emph{value}  $u_i(\vec{x}, \vec{y})$ defined as
	\[
	u_i(\vec{x}, \vec{y}) = v(x_{i-r}, \dotsc, x_{i}, y_{i-r}, \dotsc, y_{i}).
	\]
	Finally, we apply the aggregation function $\aggr$ to values $u_i$ to determine the value $u(\vec{x}, \vec{y})$ of the solution.
  That is, if the aggregation function is $\asum$,
	the cost function is given by
	\[
	   f_n(\vec{x}, \vec{y}) = \sum_{i=1}^{n} u_i(\vec{x}, \vec{y}).
	\]
  For example, if the objective is $\min$, the task in $\Pi$ is to find a solution $\vec{y}$ that minimizes $u(\vec{x}, \vec{y})$ for a given input $\vec{x}$, and so on. Note that $X$, $Y$ and $(f_n)_{n \ge 1}$ define a request-answer game.

  Note that bounded monotone minimization games, defined in Section~\ref{sec:clocked-algorithms} are not necessarily local   optimization problems.
  The latter are monotone games with finite diameter, but they do not necessarily have bounded delay.
  We emphasize that local optimization problems include all metrical task systems~\cite{Borodin1992}.

	\subparagraph{Shorthand Notation.}
  In general, the local cost function $v$ is a function with ${2(r+1)}$ arguments. However, it is often more convenient to represent $v$
  as a function that takes one matrix with two rows and $r+1$ columns and use ``$\cdot$'' to denote irrelevant parameters, e.g.\
	\[
	\vmat{\cdot&\cdot&c}{d&e&\cdot} = \gamma
	\]
	is equivalent to saying that $v(a,b,c,d,e,f) = \gamma$ for all $a,b \in X$ and $f \in Y$.

	\subsection{Encoding Examples of Online Problems}
	Let us first see how to encode typical online problems in our formalism. We start with a~highly simplified version of the \emph{online file migration problem}, a.k.a.\ \emph{online page migration} \cite{Black1989}.

	\begin{example}[online file migration]\label{ex:file-migration}
	We are given a network consisting of two nodes, and an indivisible shared resource, a~\emph{file}, initially stored at one of the nodes. Requests to access the file arrive from nodes of the network over time, and the serving cost of a request is the distance from the requesting node to the file, i.e., 0 if the file is co-located with the request, and 1 otherwise. After serving a~request, we may decide to \emph{migrate} the file to a different node of the network, paying $\costmig$ units of migration cost for some parameter $\costmig \ge 0$.

	Let us express the online file migration problem introduced earlier using the above formalism.
	The problem is modeled so that input $x_i \in X = \{ 0, 1\}$ represents access to the file at time $i$ from the node $x_i$ of the network, and output $y_i \in Y = \{0, 1\}$ represents the location of the file at time $i$.
	 We choose the horizon $r = 1$, aggregation function ``$\asum$'', and objective ``$\min$'', and define the local cost function as
		\begin{alignat*}{4}
			\vmat{\cdot&0}{0&0} &= 0, \quad &
			\vmat{\cdot&0}{1&1} &= 1, \quad &
			\vmat{\cdot&0}{1&0} &= \costmig, \quad &
			\vmat{\cdot&0}{0&1} &= 1+\costmig, \\
			\vmat{\cdot&1}{1&1} &= 0, \quad &
			\vmat{\cdot&1}{0&0} &= 1, \quad &
			\vmat{\cdot&1}{0&1} &= \costmig, \quad &
			\vmat{\cdot&1}{1&0} &= 1+\costmig.
		\end{alignat*}
		Recall that $\costmig > 0$ is the cost of migrating the file. Intuitively, the four columns represent local access, remote access, local and remote access after reconfiguration.
	\end{example}

	Let us now look at a problem of a different flavor, a variant of load balancing~\cite{Azar1994}.

	\begin{example}[online load balancing]\label{ex:load-balancing}
		Each day $i$ a job arrives; the job has a duration $x_i \in X = \{1,2,\dotsc,\ell\}$. We need to choose a machine $y_i \in Y$ that will process the job. If, e.g., $x_i = 3$, then machine $y_i$ will process job $i$ during days $i$, $i+1$, and $i+2$. The \emph{load} of a~machine is the number of concurrent jobs that it is processing at a given day, and our task is to minimize the maximum load of any machine at any point of time.

		In this case we can choose the horizon $r = \ell-1$, aggregation function ``$\max$'', and objective ``$\min$'', and define the local cost function as follows:
		\[
		v(x_{i-r},\dotsc,x_{i},y_{i-r},\dotsc,y_{i}) = \max_{y \in Y} {\bigl| \bigl\{ j \in X : x_{i-j+1} \ge j \text{ and } y_{i-j+1} = y \bigr\} \bigr|}.
		\]
		That is, we count the number of jobs that were assigned to each machine $y \in Y$ on days ${i-r}, \dotsc, i$ and that are long enough so that they are still being processed during day $i$. For example, if $X = Y = \{1,2\}$, this is equivalent to
		\[
		\vmat{1&\cdot}{\cdot&\cdot} = 1, \quad
		\vmat{2&\cdot}{2&2} = 2, \quad
		\vmat{2&\cdot}{1&1} = 2, \quad
		\vmat{2&\cdot}{1&2} = 1, \quad
		\vmat{2&\cdot}{2&1} = 1.
		\]
	\end{example}

	\subsection{Encoding Examples of Graph Problems on Paths}
	In this section, we uncover and exploit connections between time-local online algorithms and distributed graph algorithms on paths. We have seen that the formalism that we use is expressive enough to capture typical online problems; we now express some classic graph optimization problems studied in distributed computing.
  Let us now see how to express some classic graph optimization problems that have been studied in the theory of distributed computing.

  We interpret each index $i$ as a node in a path, where nodes $i$ and $i+1$ are connected by an edge. Input $x_i$ is the \emph{weight} of node $i$, and output $y_i$ encodes a subset of nodes $S \subseteq \{1,2,\dotsc,n\}$, with the interpretation that $i \in S$ whenever $y_i = 1$. Hence $X = \bbR_{\ge 0}$ and $Y = \{0,1\}$.

	\begin{example}[maximum-weight independent set]\label{ex:ind-set}
		We can capture a problem equivalent to the classic maximum-weight independent set as follows: we choose the horizon $r = 1$, aggregation function ``$\asum$'', and objective ``$\max$'', and define the local cost function as follows:
		\[
		\vmat{\cdot&\cdot}{\cdot&0} = 0, \quad
		\vmat{\cdot&\alpha}{0&1} = \alpha, \quad
		\vmat{\cdot&\cdot}{1&1} = -\infty.
		\]
		That is, a node of weight $\alpha$ is worth $\alpha$ units if we select it. The last case ensures that the solution represents a valid independent set (no two nodes selected next to each other).
	\end{example}

	\begin{example}[minimum-weight dominating set]\label{ex:dom-set}
		To represent minimum-weight dominating sets, we choose $r = 2$, $\aggr = \asum$, and $\obj = \min$. We define the local cost function as follows:
		\[
		\vmat{\cdot&\alpha&\cdot}{\cdot&1&\cdot} = \alpha, \quad
		\vmat{\cdot&\cdot&\cdot}{\cdot&0&1} = 0, \quad
		\vmat{\cdot&\cdot&\cdot}{1&0&\cdot} = 0, \quad
		\vmat{\cdot&\cdot&\cdot}{0&0&0} = +\infty.
		\]
		Here if we select a node of cost $\alpha$, we pay $\alpha$ units. Nodes that are not selected, but that are correctly dominated by a neighbor are free. We ensure correct domination by assigning an infinite cost to unhappy nodes.

		Technically, when we select a node $i$, we will pay for it at time $i+1$, not at time $i$, but this is fine, as we will in any case sum over all nodes (and ignore constantly many nodes near the boundaries).
	\end{example}

\section{Time-Local Online Algorithms vs. Local Graph Algorithms}\label{sec:temporal-vs-spatial-locality}

In this section, we discuss the connection between time-local online algorithms and local distributed graph algorithms on paths. Although the former deal with locality in the \emph{temporal} dimension and the latter in \emph{spatial} dimension, we will see that these two worlds are closely connected. In particular, we show how to transfer results from distributed computing to the time-local online setting.

We focus on two standard models with very different computational power: the anonymous port-numbering model (a weak model) and the supported \local model (a strong model).  In the deterministic setting, the correspondence between these models and time-local online algorithms is summarized in Table~\ref{table:tl-da}.

First, we now extend our study of local algorithms to cover locality in space as well.

\subsection{Local Algorithms in Time and Space}
\label{sec:formalism}

  For convenience, we will extend the definition of inputs to include a~placeholder value $\bot$ and let $x_i = \bot$ for $i < 1$ and for $i > n$. The key models of computing that we study are all captured by the following definition:
	\begin{definition}[local algorithm]\label{def:local-algorithm}
		An $[a, b]$-local algorithm is a~sequence $(A_i)_{i \ge 1}$ of functions of the form $A_i \colon X^{a+1+b} \to Y$. The output $\vec{y}$ of an algorithm $A$ for input $\vec{x} \in X^n$, in notation $\vec{y} = A(\vec{x})$, is defined as follows:
		\[
		y_i = A_i(x_{i-a}, \dotsc, x_{i+b}) \text{ for each } i = 1, \dotsc, n.
		\]
    If $A_i = A_j$ for all $i,j \in \mathbb{N}$, then the algorithm $A$ is regular. Otherwise, it is clocked.
	\end{definition}

  Note that \emph{regular} time-local algorithms as defined above are unaware of the current time step $i$; they make the same deterministic decision every time for the same (local) input pattern. We can
  quantify the cost of not being aware of the current time step, by comparing regular algorithms against the stronger model of \emph{clocked} algorithms, which can make different decisions based on the current time step $i$.

	\subparagraph{Classic Models of Online and Distributed Algorithms.}
Using the notion of \emph{regular} time-local algorithms,
we can characterize algorithms studied in prior work as follows; see also Figure~\ref{fig:overview}.
  In what follows, $T$ is a~constant independent of the length $n$ of input:
	\begin{itemize}
		\item \highlight{$[\infty,\infty]$-local:} These are algorithms with access to the full input. In the context of online algorithms, these are usually known as \emph{offline algorithms}, while in the context of distributed computing, these are usually known as \emph{centralized algorithms}.
		\item \highlight{$[\infty,-1]$-local:} These are \emph{online algorithms} in the usual sense. The output for a~time step~$i$ is chosen based on inputs for all previous time steps up to the time step $i-1$. This is an appropriate definition for the online file migration problem (Section~\ref{sec:file-migration}):
		we need to decide where to move the file before we see the next request.
		\item \highlight{$[\infty,0]$-local:} These are online algorithms with one unit of \emph{lookahead}. The output for a~time step $i$ is chosen based on inputs up to the time step $i$. This is an appropriate definition for the online load balancing problem (Example~\ref{ex:load-balancing}): we can choose the machine once we see the parameters of the new job.
		\item \highlight{$[T,T]$-local:} These can be interpreted as \emph{$T$-round distributed algorithms} in directed paths in the port-numbering model. In the port-numbering model, in $T$ synchronous communication rounds, each node can gather full information about the inputs of all nodes within distance $T$ from it, and nothing else. This is a~setting in which it is interesting to study graph problems such as the maximum-weight independent set (Example~\ref{ex:ind-set}) and the minimum-weight dominating set (Example~\ref{ex:dom-set}).
	\end{itemize}

  \subparagraph{New Models: Time-Local Online Algorithms.}
	Now we are ready to introduce the main objects of study for the present work:
	\begin{itemize}
		\item \highlight{regular $[T,-1]$-local:} These are \emph{time-local algorithms} with horizon $T$, i.e., online algorithms that make decisions based on only $T$ latest inputs.
		\item \highlight{regular $[T,0]$-local:} These are time-local algorithms with one unit of \emph{lookahead}.
	\item \highlight{clocked $[T,T]$-local:} As we will see later, these algorithms are equivalent to $T$-round distributed algorithms a~restricted variant of the \emph{supported} \local model~\cite{schmid2013exploiting,foerster2019power}.
		\item \highlight{clocked $[T,-1]$-local:} These are \emph{clocked time-local algorithms} that make decisions based on only $T$ latest inputs, but the decision may depend on the current time step $i$.
	\end{itemize}
  We note that there is nothing fundamental about the constants $-1$ and $0$ that appear above; they are merely constants that usually make most sense in applications. One can perfectly well study, e.g., $[10, 7]$-local algorithms, and interpret them either as (a) distributed algorithms that make decisions based on an asymmetric local neighborhood or (b) as time-local algorithms that can postpone decisions and choose $y_i$ only after seeing inputs~up~to~$i+7$.

\begin{table}
  \begin{center}
    \caption{Correspondence between time-local online algorithms and distributed graph algorithms.\label{table:tl-da}}
\begin{tabular}{@{}l@{\qquad}l@{\qquad}l@{}}
\toprule
& Time-local online algorithms & Local distributed graph algorithms \\
& & on directed paths \\
\midrule
Weakest & regular $[T,T]$-local & $T$ rounds in the $\pn$ model \cite{angluin80local,attiya88computing,yamashita96computing} \\
& N/A & $T$ rounds in the $\local$ model \cite{Peleg2000,linial92locality} \\
& clocked $[T,T]$-local &  $T$ rounds in the numbered $\local$ model \\
Strongest & N/A &  $T$ rounds in the supported $\local$ model \cite{schmid2013exploiting,foerster2019power} \\
\bottomrule
\end{tabular}
\end{center}
\end{table}

\subsection{Distributed Graph Algorithms}
Let $G = (V,E)$ be a graph that represents the communication topology of a distributed system consisting of $n$ nodes $V = \{v_1, \ldots, v_n\}$. Each node $v_i \in V$ corresponds to a processor and the edges denote direct communication links between processors, i.e., any pair of nodes connected by an edge can directly communicate with each other. In this work, $G$ will always be a path of length $n$ with the set of edges given by $E = \{ \{ v_i, v_{i+1} \} : 1 \le i < n \}$.

\subparagraph{Synchronous Distributed Computation.}
We start with the basic synchronous message-passing model of computation.
Let $X$ and $Y$ be the set of input and output labels, respectively, The input is the vector $\vec x = (x_1, \ldots, x_n) \in X^n$, where $x_i$ is the local input of node $v_i$. Initially, each node $v_i$ only knows its local input $x_i \in X$.

The computation proceeds in synchronous rounds, where in each round $t = 1, 2, \ldots$, all nodes in parallel perform the following in lock-step:
\begin{enumerate}
\item send messages to their neighbors,
\item receive messages from their neighbors, and
\item update their local state.
\end{enumerate}
An algorithm has running time $T$ if at the end of round $T$, each node $v_i$ halts and declares its own local output value $y_i$. The output of the algorithm is the vector $\vec y = (y_1, \ldots, y_n) \in Y^n$.

Note that---since there is no restriction on message sizes---every $T$-round algorithm can be represented as a simple full-information algorithm: In every round, each node broadcasts all the information it currently has, i.e., its own local input and inputs it has received from others, to all of its neighbors. After executing this algorithm for $T$ rounds, this algorithm has obtained all the information \emph{any} $T$-round algorithm can. Thus, every $T$-round algorithm can be represented as map from radius-$T$ neighborhoods to output values.

\subsection{Distributed Algorithms vs.\ Time-Local Online Algorithms}

The distributed computing literature has extensively studied the computational power of different variants of the above basic model of graph algorithms.
The variants are obtained by considering different types of \emph{symmetry-breaking information}: in addition to the problem specific local input $x_i \in X$, each node $v_i$ also receives some input $z_i$ that encodes additional model-dependent symmetry-breaking information.

We will now discuss four such models in increasing order of computational power. The correspondence between these models and time-local online algorithms is summarized by Table~\ref{table:tl-da}.

\subparagraph{\boldmath The port-numbering model $\pn$ on directed paths.} In the $\pn$ model \cite{angluin80local,attiya88computing,yamashita96computing} all nodes are anonymous, but the edges of $G$ are consistently oriented from $v_i$ towards $v_{i+1}$ for all $1 \le i < n$. The nodes know their degree and can distinguish between the incoming and outgoing edges.
The orientation only serves as symmetry-breaking information; the communication links are bidirectional.

Any deterministic algorithm in this model corresponds to a map $A \colon X^{2T+1} \to Y$ such that the output of node $v_i$ for $1 \le i \le n$ is
\[
y_i = A(x_{i-T}, \ldots,x_i,\ldots, x_{i+T}),
\]
where we let $x_j = \bot$ for any $j < 0$ or $j > n$ (the $\bot$ values are used in the scenarios where nodes near the endpoints of the path observe these endpoints).
Note that this is exactly the definition of a regular $[T,T]$-local algorithms (Definition~\ref{def:local-algorithm}).

\subparagraph{\boldmath The $\local$ model on directed paths.} In the $\local$ model \cite{linial92locality,Peleg2000} each node receives the same information as in the port-numbering model $\pn$, but in addition, each node $v_i$ is also given a unique identifier $\ID(v_i)$ from the set $\{ 1_, \ldots, n^c \}$ for some constant $c \ge 1$; the nodes do not know $n$. Lower bounds for this model also hold in the weaker $\pn$ model.

\subparagraph{\boldmath The numbered $\local$ model on directed paths.} The numbered $\local$ model further assumes that the unique identifiers have a specific, ordered structure: node $v_i$ is given the identifier $\ID(v_i) = i$ as local input in addition to the problem specific input $x_i \in X$. That is, each node knows its distance from the start of the path.
Any deterministic algorithm in this model corresponds to a map $A \colon X^{2T+1} \times \mathbb{N} \to Y$ such that the output of node $v_i$ for $1 \le i \le n$ is
\[
y_i = A(x_{i-T}, \ldots,x_i,\ldots, x_{i+T}, i) = A_i(x_{i-T}, \ldots,x_i,\ldots, x_{i+T}).
\]
Observe that this coincides with \emph{clocked} $[T,T]$-local algorithms (Definition~\ref{def:local-algorithm}).

This model is \emph{not} something that to our knowledge has been studied in the distributed computing literature; the name ``numbered $\local$ model'' is introduced here. However, it is very close to another model, so-called supported $\local$ model, which has been studied in the literature.

\subparagraph{\boldmath The supported $\local$ model on directed paths.} The supported $\local$ model \cite{schmid2013exploiting,foerster2019power} is the same as the numbered $\local$ model, but each node is also given the length $n$ of the path as local input. This would correspond to clocked $[T,T]$-local algorithms that also know the length of the input in advance but do not see the full input. This is the most powerful model, and hence, all impossibility results in this model also hold for all the previous models.

\subsection{Transferring Results From Distributed Computing}

\subsubsection{Symmetry-Breaking Tasks in Distributed Computing}

One of the key challenges in distributed graph algorithms is local symmetry breaking: two adjacent nodes in a graph (here: two consecutive nodes along the path) have got isomorphic local neighborhoods but are expected to produce different outputs.

In distributed computing, a canonical example is the \emph{vertex coloring problem}. Consider, for example, the task of finding a proper coloring with $k$ colors. This is trivial in the supported and numbered models (node number $i$ can simply output e.g.\ $i \bmod 2$ to produce a proper $2$-coloring). However, the case of the $\pn$ model and the $\local$ model is a lot more interesting.

One can use simple arguments based on \emph{local indistinguishability} \cite{boldi01effective,yamashita96computing} to argue that such tasks are not solvable in $o(n)$ rounds in the $\pn$ model. In brief, if two nodes have identical radius-$T$ neighborhoods, then they will produce the same output in any deterministic $\pn$-algorithm that runs in $T$ rounds. For example, it immediately follows $k$-coloring for any $k$ requires $\Omega(n)$ rounds in the deterministic $\pn$ model.

Yet another idea one can exploit in the analysis of symmetry-breaking tasks is \emph{rigidity} (or, put otherwise, the lack of \emph{flexibility}); see e.g.\ \cite{chang2020distributed,brandt2017lcl}. For example, $2$-coloring is a rigid problem: once the output of one node is fixed, all other nodes have fixed their outputs. Informally, two nodes arbitrarily far from each other need to be able to coordinate their decisions---or otherwise there is at least one node between them that produces the wrong output. This idea can be used to quickly show that e.g.\ $2$-coloring in the $\local$ model requires also $\Omega(n)$ rounds, and this holds even if we consider \emph{randomized} algorithms (say, Monte Carlo algorithms that are supposed to work w.h.p.).

This leaves us with the case of symmetry-breaking tasks that \emph{are} flexible. A canonical example is the $3$-coloring problem. Informally, one can fix the colors of any two nodes (sufficiently far from each other), and it is always possible to complete the coloring between them. While the $3$-coloring problem requires $\Omega(n)$ rounds in the deterministic $\pn$ model, it is a problem that can be solved much faster in the deterministic $\local$ model and also in the randomized $\pn$ model: the Cole--Vishkin technique \cite{cole86deterministic} can be used to do it in only $O(\log^* n)$ rounds. However, what is important for us in this work is that this is also known to be tight \cite{linial92locality,naor91lower}: $3$-coloring is not possible in $o(\log^* n)$ rounds, not even if we use both unique identifiers and randomness.

Moreover, the same holds for \emph{all} problems in which the task is to label a path with some labels from a constant-sized set $Y$, and arbitrarily long sequences of the same label are forbidden: no such problem can be solved in constant time in the $\pn$ or $\local$ model, not even if one has got access to randomness \cite{linial92locality,naor91lower,naor1995can,chang2019time,suomela13survey}.

We will soon see what all of this implies for us, but let us discuss one technicality first: symmetric vs.\ asymmetric horizons.

\subsubsection{Symmetric vs.\ Asymmetric Horizons}

While the standard models in distributed computing correspond to \emph{symmetric} horizons ($[T,T]$-local algorithms) and the study of online algorithms is typically interested in \emph{asymmetric} horizons (e.g.\ $[T,-1]$-local algorithms), in many cases this distinction is inconsequential when one considers symmetry-breaking tasks.

Consider, for example, the vertex coloring problem $\Pi$. Assume one is given an $[a,b]$-local algorithm $A$ for solving $\Pi$. Now for any constant $c$ one can construct an $[a+c,b-c]$-local algorithm $A'$ that solves the same problem. In essence, node $x_i$ in algorithm $A'$ simply outputs $A(x_{i-a-c},\dotsc,x_{i+b-c})$. Now if one compares the outputs of $A'$ and $A$, we produce the same sequence of colors but shifted by $c$ steps. This is the standard trick one uses to convert algorithms for directed paths into algorithms for rooted trees and vice versa; see e.g.\ \cite{rybicki15exact,chang2020distributed}. The only caveat is that we need to worry about what to do near the boundaries, but for our purposes the very first and the very last outputs are usually inconsequential (can be handled by an ad hoc rule, or simply ignored thanks to the additive constant in the definition of the competitive ratio).

Hence, in essence everything that we know about symmetry-breaking tasks in the context of $[T,T]$-local algorithms can be easily translated into equivalent results for $[T',-1]$-local algorithms for $T' = 2T+1$, and vice versa.

\subsubsection{Distributed Optimization and Approximation}

So far we have discussed distributed graph problems in which the task is to find \emph{any} feasible solution subject to some local constraints. However, especially in the context of online algorithms, we are usually interested in finding \emph{good} solutions. Typical examples are problems such as the task of finding the minimum dominating set problem and the maximum independent set problem.

These are not, strictly speaking, symmetry-breaking tasks. Nevertheless, it turns out to be useful to look at also such tasks through the lens of symmetry breaking. In brief, the following picture emerges \cite{suomela13survey,naor1995can,goos13local-approximation,czygrinow08fast}:
\begin{itemize}
    \item Deterministic $O(1)$-round $\local$-model algorithms are not any more powerful than deterministic $O(1)$-round $\pn$-model algorithms.
    \item Randomized $O(1)$-round algorithms are strictly stronger than deterministic $O(1)$-round algorithms.
\end{itemize}
For example, if we look at the \emph{minimum dominating set problem} in unweighted paths, the only possible deterministic $O(1)$-round $\pn$-algorithm produces a constant output: all nodes (except possibly some nodes near the boundaries) are part of the solution. Deterministic $O(1)$-round $\local$-algorithm can \emph{try} to do something much more clever, with the help of unique identifiers, but a Ramsey-type argument \cite{naor1995can,goos13local-approximation,czygrinow08fast} shows that it is futile: there always exists an adversarial assignment of unique identifiers such that the algorithm produces a near-constant output for all but $\epsilon n$ many nodes, for an arbitrarily small $\epsilon > 0$. However, randomized algorithms can do much better (at least on average); to give a simple example, consider an algorithm that first takes each node with some fixed probability $0 < p < 1$, and then adds the nodes that were not yet dominated. Finally, in the numbered and supported models one can obviously do much better, even deterministically (simply pick every third node).

This is now enough background on the most relevant results related to $T$-round algorithms in deterministic and randomized $\pn$ and $\local$ models.

\subsubsection{Consequences: Time-Local Solvability}

It turns out to be highly beneficial to try to classify online problems in the above terms: whether there is a component that is equivalent to a symmetry-breaking task or to a nontrivial distributed optimization problem. This is easiest to explain through examples:

\subparagraph{Online file migration (Example~\ref{ex:file-migration}).}

This problem is trivial to solve for a constant input; the same also holds for any input sequence that is strictly periodic. Indeed, if the adversary gives a long sequence of constant inputs (or follows a fixed periodic pattern), it only helps us. Hence none of the above obstacles are in our way; interesting inputs are sequences that already break symmetry locally. Furthermore, as we also know that this is a well-known online problem solvable with the full history, we would expect that there is also a regular time-local algorithm for solving the task, with a nontrivial competitive ratio. While this is a \emph{heuristic} argument (based on the \emph{lack} of specific obstacles), we will see in Section~\ref{sec:file-migration} that the argument works very well in this case.

\subparagraph{Online load balancing (Example~\ref{ex:load-balancing}).}

This problem is fundamentally different from the file migration problem. Let us assume that the algorithm needs to output the action (on which machine to schedule the current job). Consider an input sequence that consists of the constant value~$2$. In such a case, there is an optimal solution that alternately assigns the $2$-unit jobs to the two machines, ensuring that the load of any machine at any time is exactly $1$. But this means that an optimal algorithm has to turn the constant input $2,2,2,2,\dotsc$ into a strictly alternating sequence like $1,2,1,2,\dotsc$. Any deviation from it will result at least momentarily in a load of $2$. Hence in an \emph{optimal} solution we need to \emph{at least} solve the $2$-coloring problem within each segment of such constant inputs. As we discussed, this is not possible in the $\pn$ or $\local$ model in $O(1)$ rounds, not even with the help of randomness; it follows that there certainly is no optimal regular time-local online algorithm, with any constant horizon $T$. Optimal solutions have to resort to the clock.

However, this does not prevent us from solving the problem with a finite competitive ratio. Indeed, even the trivial solution that outputs always $1$ will result in a maximum load that is at most $2$ times as high as optimal.

Furthermore, if we were not interested in the \emph{maximum} load but the \emph{average} load, we arrive at a task that is, in essence, a distributed optimization problem. Regular \emph{randomized} time-local algorithms may then have an advantage over regular \emph{deterministic} time-local algorithms, and indeed this turns out to be the case here: simply choosing the machine at random is already better on average than assigning all tasks to the same machine.

\subsubsection{Consequences: Time-Local Models}

On a more general level, the above discussion also leads to the following observation: the definition of \emph{regular} time-local algorithms is \emph{robust}. Now it coincides with the $\pn$ model, but even if one tried to strengthen it so that its expressive power was closer to the $\local$ model, very little would change in terms of the results.

Conversely, if one weakened the \emph{clocked} model so that e.g.\ the clock values are not increasing by one but they are only a sequence of monotone, polynomially-bounded time stamps, we would arrive at a model very similar to the $\local$ model, and as we have seen above, time-local algorithms in such a model cannot solve symmetry-breaking tasks any better than in the regular model. Hence in order to capture the idea of a model that is strictly more powerful than the regular model, it is not sufficient to have a definition in which the clock values are merely monotone and polynomially bounded, but one has to further require e.g.\ that the clock values increase at each step at most by a constant. (Such a model with constant-bounded clock increments would indeed be a meaningful alternative, and it would fall in its expressive power strictly between our regular and clocked models. It would be strong enough to solve $3$-coloring but not strong enough to solve $2$-coloring in a time-local fashion. We do not explore this variant further, but it may be an interesting topic for further research, especially when comparing its power with randomized $\pn$ algorithms.)

\section{Randomized Local Algorithms}
\label{sec:random}

In this section, we define randomized time-local online algorithms.
We use these definitions in our studies of the online file migration problem: in Section~\ref{sec:synthesis} we synthesize behavioral time-local algorithms, and in Section~\ref{sec:file-migration} we establish a lower bound (Theorem~\ref{thm:file-migration-lb}) for randomized \emph{clocked} time-local algorithms, which studies degradation of the competitive ratio with decreasing $T$.

In the classic online setting, there are two equivalent ways of describing randomized algorithms:
\begin{itemize}
\item at the start, randomly sample an algorithm from a~set of deterministic algorithms, or
\item at each step, make a random decision based on coin flips.
\end{itemize}
The former corresponds to \emph{mixed strategies}, where we sample all random bits used by the algorithm before seeing any of the input, whereas the latter corresponds to \emph{behavioral strategies}, where the algorithm generates random bits along the way as it needs them.

\subparagraph*{Mixed vs.\ Behavioral Strategies in Time-Local Algorithms.}
The above two characterizations are equivalent in classic online algorithms~\cite{Borodin1998}: to simulate a~behavioral strategy with a~mixed strategy, we can generate an infinite sequence $(r_i)_{i \ge 1}$ of random bit strings in advance and use the random bits given by $r_i$ in step $i$. Conversely, we can choose to flip coins only at the beginning and store the outcomes in memory and refer to them consistently at later steps.

In contrast, for time-local algorithms, the behavioral and mixed strategies differ in a~way we can exploit randomness, and each type of strategy brings distinct advantages. If we use a behavioral strategy, at each step the algorithm can make coin flips that are independent of the previous coin flips. This enables algorithmic strategies that can e.g.\ break ties in an independent manner in successive steps.
If we use a mixed strategy, we commit to a~randomly chosen (consistent) strategy: the initial random choice influences all outputs.
Interpreting the differences between the two types of randomness in terms of distributed models~\cite{NewmanS96},  behavioral time-local strategies correspond to \emph{private randomness} available at each step~$i$, whereas
mixed time-local strategies correspond to \emph{shared randomness} across the whole sequence. Interestingly, in the time-local setting, it is also natural to consider a~\emph{combination of both}: we choose a behavioral time-local strategy at random.

With this in mind, we arrive at three natural definitions of randomized time-local algorithms:
\begin{enumerate}
\item \emph{Behavioral strategy} time-local algorithms,
\item \emph{Mixed strategy} time-local algorithms,
\item \emph{General strategy} time-local algorithms that use a combination of both.
\end{enumerate}
We now give formal definitions for each class of randomized time-local algorithms.

\begin{definition}[behavioral local algorithms]
  A behavioral $[a,b]$-local algorithm is given by the sequence of maps $(A_i)_{i \ge 1}$ of the form
 $A_i \colon X^{a+b} \times [0,1) \to Y$, where the output is given by
\[
y_i = A_i(x_{i-a}, \ldots, x_{i+b}, r_i),
\]
where $(r_i)_{i \ge 1}$ is a sequence of i.i.d.\ real values sampled uniformly from the unit range.
If $A_i=A_j$ for all $i,j$, then the algorithm is regular. Otherwise, it is clocked.
\end{definition}

\begin{definition}[mixed local algorithms]
  \label{def:rand-with-clock}
  Let $\mathcal{D}$ be a nonempty set of (deterministic) \mbox{$[a,b]$-local} algorithms.
  A mixed $[a,b]$-local algorithm over $\mathcal{D}$ is a probability measure $A \colon \mathcal{D} \to [0,1]$ over $\mathcal{D}$. The output of $A$ on input $\vec x$ is the random vector $\vec y = P(\vec x)$, where $P$ is a~deterministic time-local algorithm sampled from $\mathcal{D}$ according to $A$.
  If $\mathcal{D}$ is a subset of all regular $[a,b]$-local algorithms, then $A$ is regular. If $\mathcal{D}$ is a subset of all clocked $[a,b]$-local algorithms, then $A$ is clocked.
\end{definition}

\begin{definition}[general randomized local algorithms]
  A general randomized regular $[a,b]$-local algorithm is a mixed $[a,b]$-local algorithm over the set of regular behavioral $[a,b]$-local algorithms.
  A general randomized clocked $[a,b]$-local algorithm is a mixed $[a,b]$-local algorithm over the set of clocked behavioral $[a,b]$-local algorithms.
\end{definition}

\begin{theorem}
The class of general clocked randomized time-local algorithms is equivalent to mixed clocked time-local algorithms.
\label{thm:random-clock}
\end{theorem}
\begin{proof}
The general randomized time-local algorithms can be simulated by the mixed clocked algorithms: we can generate an infinite sequence $(r_i)_{i \ge 1}$ of random bit strings in advance and store them in functions $A_i$ of the deterministic clocked $[a,b]$-local algorithms, and use the random bits given by $r_i$ in step $i$. On the other hand, the mixed clocked time-local algorithms are contained in general randomized time-local algorithms, which concludes our claim.
\end{proof}

However, note that Theorem~\ref{thm:random-clock} does not hold for regular algorithms:
as time-local algorithms do not have memory to store past random outcomes, it is impossible to directly simulate mixed time-local algorithms by a behavioral time-local algorithm that flips coins only at the beginning.

The role of randomization was merely scratched in this work. We established that with clock, mixed time-local algorithms are at least as powerful as behavioral time-local algorithms.
As behavioral time-local algorithms cannot store past random coins, their power seems limited. Determining the relations between types of randomness is left to future work.

\subparagraph*{Adversaries and the Expected Competitive Ratio.}
We naturally extend the notion of competitiveness of time-local algorithms to randomized algorithms. For randomized algorithms, the answer sequence and the cost of an algorithm is a random variable. We will abuse the notation slightly to let $\vec y = A(\vec x)$ denote the random output generated by a randomized algorithm $A$ on input $\vec x$.

We say that a randomized online algorithm $A$ for a game defined with cost functions $(f_n)_{n \ge 1}$ is \emph{$c$-competitive} if
\[
\mathbb{E}[f_n(\vec x, A(\vec x))] \le c \cdot \OPT(\vec x)+d
\]
for any input sequence $\vec x$ and a fixed constant $d$.
The input sequence and the benchmark solution $\OPT$ is generated by an adversary.
We distinguish between the notion of competitiveness against various adversaries, having different knowledge about $A$ and different knowledge while producing the solution $\OPT$.
Competitive ratios for a~given problem may vary depending on the power of the adversary.
The adversary model used in this paper is an~\emph{oblivious offline adversary}, who must produce an input sequence in advance, merely knowing the description of the algorithm it competes against (in particular, it may have access to probability distributions that the algorithm uses, but not the random outcomes), and pays an optimal offline cost for the sequence.
For a~comprehensive overview of adversary types, see \cite{Borodin1998}.

We raise a question regarding the adaptive offline adversary in the time-local setting.
A well-known result in classic online algorithms states that if there exists a \mbox{$c$-competitive} randomized algorithm against it, then there exists a deterministic \mbox{$c$-competitive} algorithm, for any $c$~\cite{Ben-David1994}.
Does the existence~of a~competitive randomized time-local algorithm against the adaptive offline adversary imply the existence of any competitive deterministic \mbox{\emph{time-local}} algorithm?

\section{Automated Algorithm Synthesis}\label{sec:synthesis}

In this section, we describe a technique for automated design of time-local algorithms for local optimization problems, defined in Section~\ref{sec:local-problems}. This technique allows us to automatically obtain both upper and lower bounds for regular time-local algorithms. In particular, for deterministic algorithms, we can synthesize \emph{optimal} algorithms. We also discuss how to extend our approach to randomized algorithms.
As our case study problem, we use the simplified variant of online file migration.

\subsection{Overview of the Approach}

We now assume that the input and output sets $X$ and $Y$ are finite. Recall that a regular time-local algorithm that has access to last $T$ inputs is given by a map $A \colon X^T \to Y$. The synthesis task is as follows: given the length $T \in \mathbb{N}$ of the input horizon, find a map $A$ that minimizes the competitive ratio.
For simplicity of presentation, we will ignore short instances of length $n < T$, as short input sequences do not influence the competitive ratio.

\subparagraph*{The Synthesis Method.}
The high-level idea of our synthesis approach is simple:
\begin{enumerate}
\item Iterate through all of the algorithm candidates in the set $\mathcal{A} = \{ X^T \to Y \}$.

\item Compute the competitive ratio $c(A)$ for each algorithm $A \in \mathcal{A}$.

\item Choose the algorithm $A$ that minimizes the competitive ratio.
\end{enumerate}
Given that the input and output sets $X$ and $Y$ are finite, the set $\mathcal{A}$ of algorithms is also finite: there are exactly $|Y|^{|X|^T}$ algorithms we need to check.

\subparagraph*{Evaluating the Competitive Ratio.}
Obviously, the challenging part is implementing the second step, i.e., computing the competitive ratio of a given algorithm $A$. A priori it may seem that we would need to consider infinitely many input strings in order to determine the competitive ratio of the algorithm. However, for any local optimization problem $\Pi$ with finite input and output sets, it turns out that we can capture the competitive ratio by analyzing a finite combinatorial object.

We show that for any time-local algorithm $A$, we can construct a (finite) weighted, directed graph $G(\Pi, A)$ that captures the costs of output sequences as walks in $G(\Pi,A)$. The cost of any regular time-local algorithm on adversarial input sequences can be obtained by evaluating the weight of all cycles defined in this graph $G(\Pi, A)$.

\subsection{Evaluating the Competitive Ratio of an Algorithm}

We will now describe how to construct the graph $G(\Pi,A)$ for a given local optimization problem $\Pi$ and a regular local algorithm $A$. For the sake of simplicity, we only consider the sum aggregation function; the construction for $\min$ and $\max$ aggregation is defined analogously.

Let $r \in \mathbb{N}$ be the horizon of the local optimization problem $\Pi$, $v$ the local cost function of $\Pi$, and $A \colon X^T \to Y$ be the regular time-local algorithm. To avoid unnecessary notational clutter, we describe the construction for $r=1$; however, the construction is straightforward to generalize.

\subparagraph*{The Dual de Bruijn Graph.}

We construct a directed graph $G = (V,E)$ on the set of vertices $V = X^T \times Y$. For any $\vec x = (x_1, \ldots, x_k)$, we define $s(\vec x, a) = (x_2, \ldots, x_k, a)$ to be the successor of $\vec x$ on~$a$. For each vertex $(\vec a, y) \in V$, there is a directed edge towards the vertex $(\vec a', y') \in V$, where for all $y' \in Y$, $\vec a' = s(\vec a, x)$ and $x \in X$. Note that there are self-loops in this graph.

The idea is that for any sufficiently long input $n \ge T$, an input sequence $\vec x \in X^n$ and an output sequence $\vec y \in Y^n$ define a walk $\rho(\vec x, \vec y)$ in the graph~$G$. After the time step $i \ge T$, we are at vertex $( x_{i-T+1}, \ldots, x_{i},  y_{i} ) \in V$ and the next vertex is given by $(x_{i-T+2}, \ldots, x_{i+1}, y_{i+1}) \in V$. In particular,
from any walk $\rho$ we can obtain the following sequences:
\begin{itemize}
\item an input sequence $\vec x(\rho) = (x_1, \ldots, x_n) \in X^n$,
\item some (possibly optimal) solution $\vec y^*(\rho) = (y_{1}, \ldots, y_{n})$ for $\vec x(\rho)$, and
\item the output $y(\rho) = A(\vec x(\rho))$ given by the algorithm on $\vec x(\rho)$.
\end{itemize}
Vice versa, any pair of input $\vec x$ and output $\vec y^*$ sequences defines a walk $\rho(\vec x, \vec y^*)$ in $G$.

\subparagraph*{Assigning the Costs.}
For each edge $e \in E$ in the graph, we assign \emph{two} costs for the edge: the first describes the cost paid by some (possibly optimal) output, and the second, the cost paid by the algorithm $A$.
Recall that for a local optimization problem $\Pi$, the costs are given by the local cost function $v \colon X^{r+1} \times Y^{r+1} \to \mathbb{R} \cup \{\infty\}$. For the case $r=1$, the function $v$ takes $4$ parameters.

Consider an edge $e =((\vec a, b), (\vec a',  b')) \in E$, where $\vec a' = (a_2, \ldots, a_{T}, x)$ for some $x \in X$.
We now define the \emph{adversary cost} $w(e)$ and \emph{algorithm cost} $q(e)$ of the edge $e$.
We define
\begin{align*}
  w(e) &= v( a_{T}, x, b, b') \textrm{ is the cost paid output } b' \textrm{ on input } x, \\
  q(e) &=  v( a_{T}, x, b, A(\vec a)) \textrm{ is the cost paid by the output of the algorithm on input } x,
\end{align*}
where $v \colon X \times Y \to \mathbb{R} \cup \{ \infty\}$ is the local cost function of the problem $\Pi$.

We note that the costs generalize to arbitrary $r>1$ by applying the definitions of local cost functions given in Section~\ref{sec:local-problems} and extending the set of vertices to be $V = X^{T+r} \times Y^r$ to accommodate the larger horizon used for the local cost function.

\subparagraph{The Cost Ratio of a Walk.}
Finally, for any walk $\rho = (v_1, \ldots, v_k)$ in $G$, we define
\[
w(\rho) = \sum_{i=1}^{k-1} w(v_i, v_{i+1}), \qquad q(\rho) = \sum_{i=1}^{k-1} q(v_i, v_{i+1}).
\]
Here $w(\rho)$ and $q(\rho)$ define the \emph{total} adversary and algorithm costs for the walk $\rho$. The \emph{cost ratio} of a walk $\rho$ is defined as
\[
r(\rho) = \begin{cases}
  q(\rho) / w(\rho) & \textrm{if } w(\rho) > 0 \\
  1 & \textrm{if } q(\rho)=w(\rho)=0 \\
  \infty & \textrm{otherwise.}
  \end{cases}
\]
That is, on input $\vec x(\rho)$ the algorithm $A$ will pay a cost of $q(\rho) + O(1)$, whereas the optimum solution has cost at most $w(\rho) + O(1)$; there is a constant overhead on the costs since we ignore the costs incurred during the first $T-1 + r = O(1)$ inputs.

\subparagraph*{Bounding the Competitive Ratio.} We now show that we can compute the competitive ratio of the algorithm $A$ using the graph $G = G(\Pi, A)$. We say a walk $\rho = (v_1, \ldots, v_k) \in V^k$ is \emph{closed} if its starts and ends in the same vertex $v_1=v_k$. A \emph{directed cycle} is a closed walk that is non-repeating, i.e., $v_i \neq v_j$ for all $1 \le i \le j < k$.

\begin{theorem}\label{thm:synthesis}
  The competitive ratio of the algorithm $A$ is 
  $$\max \{ r(\rho) : \rho \textrm{ is a directed cycle of } G \}.$$
\end{theorem}

Figure~\ref{fig:synthesis-graph} gives an example of the dual de Bruijn graph for the online file migration problem (Example~\ref{ex:file-migration}) and an algorithm with local horizon $T=2$.

To prove the above theorem, we introduce three lemmas and the following definitions. A \emph{closed extension} of a walk $\rho$ is a closed walk $\rho'$ that contains $\rho$ as a prefix. A \emph{subwalk} $\rho'$ of a walk $\rho = (v_1, \ldots, v_k)$ is a subsequence $(v_i, \ldots, v_j)$ for some $1 \le i \le j \le k$. A \emph{decomposition of $\rho$} into $L$ subwalks is a sequence of subwalks $\rho_1, \ldots, \rho_L$ of $\rho$ such that their concatenation $\rho = \rho_1 \cdots \rho_L$.

\begin{lemma}\label{lemma:walk-decomposition}
  Let $\rho$ be a walk in $G$. For any decomposition of $\rho$ into $L$ subwalks $\rho_1 \cdots \rho_L$, there exists some $1 \le i \le L$ such that $r(\rho_i) \ge r(\rho)$.
\end{lemma}
\begin{proof}
  Let $\pi$ be a permutation on $\{1, \ldots, L\}$ and $\tau_i = \rho_{\pi(i)}$ such that
  \[
   r(\tau_1) \le r(\tau_2) \le \cdots \le r(\tau_L).
   \]
   Moreover, for $1 \le i \le L$ we define $r(\tau_i) = q_i/w_i$, where $q_i = q(\tau_i)$ and $w_i = w(\tau_i)$.
   We use the shorthands $Q(i) = \sum_{j=1}^i q_j$ and $W(i) = \sum_{j=1}^i w_j$. Note that for the aggregate cost ratio for a local optimization problem using the sum as its aggregation function gives that
   \[
r(\rho) = \frac{Q(L)}{W(L)} = \frac{\sum_{j=1}^i q_j}{\sum_{j=1}^i w_j}.
   \]
   We now show by induction that for all $1 \le i \le L$ we have that
   \[
   r(\tau_i) = \frac{q_i}{w_i} \ge \frac{Q(i)}{W(i)}.
   \]
   Observe that this implies that $r(\rho_{\pi(L)}) = r(\tau_L) \ge Q(L)/W(L) = r(\rho)$.

   The base case $i=1$ is vacuous. For the inductive step, assume that the claim holds for some $1 \le i < L$.
   For the sake of contradiction, assume that claim does not hold for $i+1$, i.e.,
   \[
   r(\tau_{i+1}) = \frac{q_{i+1}}{w_{i+1}} < \frac{Q(i+1)}{W(i+1)}.
   \]
   By rearranging the terms, we get
   \[
   \frac{Q(i+1) w_{i+1} - W(i+1) q_{i+1}}{W(i+1)} > 0,
   \]
   which in turn implies that $Q(i+1) \cdot w_{i+1} > W(i+1) \cdot q_{i+1}$ holds. Now observing that
   \[
   Q(i)w_{i+1} + q_{i+1} w_{i+1} = Q(i+1) w_{i+1} > W(i+1) \cdot q_{i+1} = W(i) q_{i+1} + w_{i+1} q_{i+1},
   \]
   we get that
   \[
   r(\tau_{i+1}) = \frac{q_{i+1}}{w_{i+1}} < \frac{Q(i)}{W(i)} \le \frac{q_i}{w_i} = r(\tau_i),
   \]
   where the second inequality follows from the induction assumption. However, this contradicts the fact that $\tau_1, \ldots, \tau_L$ were ordered according to increasing cost ratio.
\end{proof}

\begin{lemma}\label{lemma:cycle-lb}
  Let $\rho$ be a directed cycle in $G$. The competitive ratio of $A$ is at least $r(\rho)$.
\end{lemma}
\begin{proof}
Recall that the cycle defines an input sequence $\vec x = \vec x(\rho)$. By definition, the algorithm has cost at least $q(\rho)$ on this input sequence, whereas the optimum solution has cost at most $w(\rho) + d$ for some constant $d$.  Thus, the algorithm has a cost of at least $q(\rho) \ge r(\rho) (w(\rho) + d) \ge r(\rho) \cdot \OPT(\vec x) + O(1)$.
\end{proof}

\begin{lemma}\label{lemma:cycle-ub}
If the competitive ratio of $A$ is greater than $c + \varepsilon$ for some $\varepsilon > 0$, then there exists a directed cycle $\rho$ in $G$ with cost ratio $r(\rho) > c$.
\end{lemma}
\begin{proof}
  For any given walk $\rho$ in $G$, let $\hat{\rho}$ be the shortest closed extension of $\rho$ that minimizes the cost of $\vec y^*(\hat{\rho})$. Note there may be multiple shortest closed extensions, so we pick one with the cheapest adversarial cost. We let $\hat{\rho} \setminus \rho$ denote the suffix of $\hat{\rho}$ that satisfies $\hat{\rho} = \rho \cdot (\hat{\rho} \setminus \rho)$.
  Define
  \[
  \delta = \max \{ w(\hat{\rho} \setminus \rho) : \rho \textrm{ is a walk in } G \}.
\]
Note that $\delta$ is a constant, since $\hat{\rho}$ is a minimal closed extension of $\rho$ and $G$ is finite.

Let $\vec x$ be an input sequence and $\vec y^*$ an \emph{optimal} output sequence. For the walk $\rho = \rho(\vec x, \vec y^*)$, we have that
\[
r(\hat{\rho}) = \frac{q(\hat{\rho})}{w(\hat{\rho})} \ge \frac{q(\rho) + q(\hat{\rho} \setminus \rho)}{w(\rho) + w(\hat{\rho} \setminus \rho)} \ge
\frac{q(\rho) + \delta}{w(\rho) + \delta},
\]
since $w(\rho) \le q(\rho)$, as the cost of the algorithm is never less than the cost of the optimal solution $\vec y^*$.
Asymptotically, as the length of the walk goes to infinity, we have that $r(\hat{\rho}) = r(\rho)  - o(1)$. In particular, for any constant $\varepsilon_0 > 0$ we can find $n_0$ such that all input sequences $\vec x$ of length $n \ge n_0$, the walk $\rho = \rho(\vec x, \vec y^*)$ given by $\vec x$ and the optimal output sequence $\vec y^*$, satisfies
\[
r(\hat{\rho}) \ge r(\rho) - \varepsilon_0.
\]
By assumption $A$ had a competitive ratio of at least $c + \varepsilon$.
We can pick a sufficiently long input sequence $\vec x$ and an optimal solution $\vec y^*$ such that $\rho = \rho(\vec x, \vec y^*)$ satisfies
\begin{align*}
  r(\hat{\rho}) &\ge r(\rho) - \varepsilon_0  \ge \frac{f_n(\vec x, A(\vec x))}{\OPT(\vec x)} - \varepsilon' - \varepsilon_0 \ge c - \varepsilon' - \varepsilon_0 + \varepsilon > c,
\end{align*}
where $f_n(\vec x, A(\vec x))$ denotes the cost of the algorithm $A$ on input $\vec x$ and $\varepsilon_0$ and $\varepsilon'$ are appropriately chosen constants.
Thus, we have now obtained a closed walk $\hat{\rho}$ with $r(\hat{\rho}) \ge c$. Since we can decompose $\hat{\rho}$ into a sequence $\hat{\rho}_1, \ldots, \hat{\rho}_K$ of directed cycles, by applying Lemma~\ref{lemma:walk-decomposition} we get that some directed cycle $\hat{\rho}_i$ satisfies $r(\hat{\rho}) > c$, as claimed.
\end{proof}

\subparagraph*{Proof of Theorem~\ref{thm:synthesis}.}
The above two lemmas yield that the competitive ratio of $A$ is
\begin{itemize}
\item at least as large as the cost-ratio of some directed cycle in $G$ (Lemma~\ref{lemma:cycle-lb}), and
\item at most as large as the cost-ratio of some directed cycle in $G$ (Lemma~\ref{lemma:cycle-ub}).
\end{itemize}
Thus, the directed cycle with the highest cost-ratio determines the competitive ratio of the algorithm $A$.
Since the graph $G$ is finite, it suffices to check all directed cycles of $G$ to determine the competitive ratio of $A$.

\subsection{Synthesis Case Study: Online File Migration}
\label{ssec:synthesis-case-study}

We now consider the case study problem of online file migration with $X = Y = \{0,1\}$ and $\costmig > 0$. Recall that Figure~\ref{fig:synthesis-graph} gives an example of graph $G$ for this problem for $T=2$. First, we discuss some optimizations and extensions to the synthesis of randomized algorithms. Finally, we overview results obtained using the synthesis framework, including optimal synthesized algorithms (cf. Figure~\ref{fig:result-overview}).

\subsubsection{Optimizations}

We discuss a few techniques for optimizing the synthesis for our case study problem of online file migration.
We can reduce the amount of computation needed to find the best algorithm $A$ for a fixed $T$ and $\costmig$, by eliminating some algorithms. For example, we can often quickly identify some simple property of $G$ that immediately disqualifies an algorithm candidate.

\subparagraph*{The Role of Self-Loops.}
If the competitive ratio of $A$ is $K$, then the cost-ratio of any directed cycle has to be at most $K$.
In particular, the cost-ratio of any directed cycle has to be finite. So we can directly eliminate all
cases in which there is a cycle $\rho$ with adversary-cost $w(\rho)=0$ and a positive algorithm-cost $q(\rho)>0$.
For example, we can apply this reasoning to self-loops in the graph $G$. If the adversary-cost of a self-loop is zero, then the algorithm-cost of the same loop has to be also zero. It follows that e.g.~we must have $A(0,\ldots,0) = 0$ and $A(1,\ldots,1) = 1$ for any algorithm $A$. In the case of $T=3$, this reduces the number of algorithms that need to be checked from $2^8=256$ to only $2^{2^3 - 2} = 2^{6} = 64$ instead.

\subparagraph*{Detecting Heavy Cycles.}
When searching for algorithms with best competitive ratio, it is useful to keep track of the best cost-ratio found so far: when checking a new algorithm candidate $A$ and its corresponding graph $G$, we can first check small cycles of length at most $L$ to see if any such cycle has cost-ratio larger than the best found cost-ratio for any other algorithm so far.
If we encounter a cycle $\rho$ with
cost-ratio $r(\rho)$ that is larger or equal than the competitive ratio of some previously
considered algorithm $A'$, then we know that the competitive ratio of $A$ is larger or equal than that
of $A'$. Thus, we can immediately disregard $A$ and move on to check the next possible algorithm candidate.

Indeed, it turns out that in many cases, cycles with large cost-ratio are already found when examining only short cycles.
However, if high-cost short cycles are not found, we can always fall back to an exhaustive search that checks all cycles.

\subsubsection{On the Synthesis of Randomized Algorithms}

We note that we can extend our approach to the synthesis of randomized algorithms (see Section~\ref{sec:random}). The synthesis bounds the \emph{expected} competitive ratio of the algorithm against an oblivious randomized adversary.
Following the distinction discussed in Section~\ref{sec:random}, we consider the synthesis for randomized \emph{behavioral} algorithms (cf.~Section~\ref{sec:random}).
Synthesis of \emph{mixed} algorithms would correspond to finding a good probability distribution over the finite set of algorithms, but we restrict our attention to the behavioral algorithms.

In the case of deterministic algorithms, we considered maps $A \colon \{0,1\}^T \to \{0,1\}$. Now we consider maps $A \colon \{0,1\}^T \to [0,1]$, where $A(\vec a)$ gives the \emph{probability} that $A$ outputs~1 upon seeing the sequence $\vec a \in X^T$ of last $T$ inputs. Thus,
\begin{align*}
   A(\vec a) &= \Pr[A \textrm{ outputs } 1 \textrm{ on input } \vec a \in X^T] \\
  1-A(\vec a) &= \Pr[A \textrm{ outputs } 0 \textrm{ on input } \vec a \in X^T].
\end{align*}
We assign the algorithm cost $q(e)$ for any edge $(\vec a, y)$ to $(s(\vec a, x), y')$ as follows:
\begin{itemize}

\item On a mismatch, the algorithm pays the cost
  \[
  q_{\textrm{mismatch}}(e) = \begin{cases}
    1- A(\vec a) & \textrm{if } x = 1 \\
    A(\vec a) & \textrm{otherwise.}
  \end{cases}
  \]

\item The switching cost is given by
  \[
  q_{\textrm{switch}}(e) = \costmig \cdot [A(\vec a) \cdot (1-A(s(\vec a, a))) + (1-A(\vec a))\cdot A(s(\vec a, x))].
  \]

  \item The total cost is $q(e) = q_{\textrm{mismatch}}(e) + q_{\textrm{switch}}(e)$.
\end{itemize}
We calculate the adversary-cost in the same manner as we do in the deterministic model.
That is our adversary always outputs 0 or 1 (but not, e.g.~$0.6$).
Thus, the graph $G$ will have the same structure as in the deterministic case.

Since there are uncountably many possible randomized algorithms $A$ for any $T$, we instead discretize the probability space into finitely many segments. Thus, we cannot guarantee that we find optimal randomized algorithms. Nevertheless, this method can be used to obtain synthesized algorithms that beat the deterministic algorithms.

\subsection{Synthesis Results}

We now give some results for the online file migration problem obtained using the synthesis approach.
We showcase time-local algorithms with small values of $T=1,2,3$ and then provide observations for $T=4$ and $T=5$.

\subsubsection{Synthesized Algorithms for \texorpdfstring{\boldmath $T = 1, 2, 3$}{T = 1, 2, 3}}
Table~\ref{table:synthesis-results} summarizes results for $T = 1, 2, 3$ and $0.1 \leq \costmig \leq 1.6$. For deterministic algorithms, we list the competitive ratios of the \emph{optimal} deterministic algorithms for the given values of parameters $T$ and $\costmig$.

For randomized algorithms, we list the best competitive ratios found by the synthesis method for the given values of $T$ and $\costmig$.  As discussed, the search for randomized algorithms was conducted in a discretized search space, so it is possible that some randomized algorithms with better competitive ratios may have been missed by the search method.

\subparagraph*{The Power of Randomness.}
Note that already with $T = 2$ we can obtain algorithms with strictly better competitive ratios when randomness is used.
Moreover, with only $T = 3$, we are able to obtain randomized algorithms with competitive ratio $< 3$ (e.g., when $\costmig = 1.0$). This is strictly better than \emph{any} (non-time-local) deterministic algorithm for $\costmig=1$.
Table~\ref{table:rando-algo} gives an example of such an algorithm that achieve competitive ratio of roughly 2.67 for $T=3$ and $\costmig = 1$.

After checking all cycles in the constructed dual de Bruijn graph, the cycle with
  the maximum cost-ratio (cost-ratio of about 2.67) happens to be the following:
\begin{itemize}
\item Last $T$ inputs: 000, adversary output: 0.
\item Last $T$ inputs: 001, adversary output: 0.
\item Last $T$ inputs: 011, adversary output: 0.
\item Last $T$ inputs: 110, adversary output: 0.
\item Last $T$ inputs: 100, adversary output: 0.
\end{itemize}

\subsubsection{The Case of \texorpdfstring{\boldmath $T = 4$}{T = 4}}

For $T=4$ we can obtain better deterministic algorithms than with $T=3$. Interestingly, we can find several optimal algorithms for the case $\costmig=1$: even a full-history deterministic online algorithm cannot achieve a better competitive ratio. Table~\ref{table:3competitive-algos} lists all the 3-competitive algorithms that exist for parameter values of $T=4$, $\costmig=1$. This shows that even very simple time-local algorithms can perform well compared to classic online algorithms. Table~\ref{table:synthesis-results} contains some of the results for $T = 4$ and $0.1 \leq \costmig \leq 1.5$.

\subsubsection{Negative Results for \texorpdfstring{\boldmath $T = 5$}{T = 5}}
Since the number of cycles to be checked increases exponentially in $T$,
we were not able to obtain any positive results for the case of $T = 5$. However,
negative results could still be obtained, since verifying for a certain lower bound
does not require to check all the cycles for all the algorithms. Instead, it is
sufficient to find at least one cycle with a large enough cost-ratio to
disregard a certain algorithm and move on to the next one. We get the following results:

\begin{observation}
  With parameter values of $\costmig = 1$ and $T = 5$, the best competitive
  ratio remains $3$. That is, for each deterministic algorithm, after
  the dual de Bruijn graph has been constructed, there is a cycle with a
  cost-ratio of at least $3$.
\end{observation}

\begin{observation}
  There is no algorithm with ratio $< 3.1$ for $T = 5$ and $\costmig = 1.1$.
\end{observation}

\begin{observation}
  There is no algorithm with ratio $< 3.2$ for $T = 5$ and $\costmig = 1.2$.
\end{observation}

\begin{table}
  \caption{The best competitive ratios for some values of $\costmig$ and $T$; see also Figure~\ref{fig:result-overview}.}
  \label{table:synthesis-results}
  \centering
  \begin{tabular}{@{}lllllll@{}}
    \toprule
    $\costmig$ & $T=1$ & $T=2$ && $T=3$ && $T=4$ \\
    \cmidrule(lr){2-2}
    \cmidrule(lr){3-4}
    \cmidrule(l){5-6}
    \cmidrule(l){7-7}
    & deterministic & deterministic & randomized & randomized & randomized & deterministic \\
    \midrule
    0.1            & 11           & 11           & 11                  & 11           & 11             & 11          \\
    0.2            & 6            & 6            & 6                   & 6            & 6              & 6           \\
    0.3            & 4.333        & 4.333        & 4.333               & 4.333        & 4.333          & 4.333       \\
    0.4            & 3.5          & 3.5          & 3.5                 & 3.5          & 3.5            & 3.5         \\
    0.5            & 3            & 3            & 3                   & 3            & 3              & 3           \\
    0.6            & 3.2          & 3.2          & 3.006               & 3.2          & 2.934          & 3.2         \\
    0.7            & 3.4          & 3.4          & 3.055               & 3.4          & 2.864          & 3.4         \\
    0.8            & 3.6          & 3.6          & 3.2                 & 3.6          & 2.797          &             \\
    0.9            & 3.8          & 3.8          & 3.35                & 3.8          & 2.734          & 3.222       \\
    \midrule
    1.0            & 4            & 4            & 3.5                 & 4            & 2.672          & 3           \\
    1.1            & 4.2          & 4.2          & 3.65                & 4.2          & 2.772          & 3.1         \\
    1.2            & 4.4          & 4.4          & 3.8                 & 4.4          & 2.872                        \\
    1.3            & 4.6          & 4.6          & 3.95                & 4.6          & 2.986          & 3.3         \\
    1.4            & 4.8          & 4.8          & 4.1                 & 4.8          & 3.088                        \\
    1.5            & 5            & 5            & 4.25                & 5            & 3.188          & 3.5         \\
    1.6            & 5.2          & 5.2          & 4.4                 & 5.2          & 3.288                        \\
    \bottomrule
  \end{tabular}
\end{table}

\begin{table}
  \caption{A randomized algorithm for $T=3$ and $\costmig=1$ with expected competitive ratio $\approx 2.67$.}
  \label{table:rando-algo}
  \centering
  \begin{tabular}{@{}l@{\qquad}lll@{}}
    \toprule
    Last $T$ inputs & The probability to output 1\\
    \midrule
    \ldots 000 & 0 \\
    \ldots 001 & 0.3309 \\
    \ldots 010 & 0.2711 \\
    \ldots 011 & 1 \\
    \ldots 100 & 0 \\
    \ldots 101 & 0.7289 \\
    \ldots 110 & 0.6691 \\
    \ldots 111 & 1 \\
    \bottomrule
  \end{tabular}
\end{table}

\begin{table}
  \caption{Three 3-competitive algorithms for $T = 4$, $\costmig = 1$.}
  \label{table:3competitive-algos}
  \centering
  \begin{tabular}{@{}l@{\qquad}lll@{}}
    \toprule
    Last $T$ inputs & \multicolumn{3}{@{}l@{}}{Output} \\
    \cmidrule{2-4}
    & $A_1$ & $A_2$ & $A_3$ \\
    \midrule
    \ldots 0000                   & 0           & 0           & 0           \\
    \ldots 0001                   & 0           & 0           & 0           \\
    \ldots 0010                   & 0           & 0           & 0           \\
    \ldots 0011                   & 1           & 1           & 1           \\
    \ldots 0100                   & 0           & 0           & 0           \\
    \ldots 0101                   & 0           & 0           & 1           \\
    \ldots 0110                   & 1           & 1           & 1           \\
    \ldots 0111                   & 1           & 1           & 1           \\
    \ldots 1000                   & 0           & 0           & 0           \\
    \ldots 1001                   & 0           & 0           & 0           \\
    \ldots 1010                   & 1           & 0           & 1           \\
    \ldots 1011                   & 1           & 1           & 1           \\
    \ldots 1100                   & 0           & 0           & 0           \\
    \ldots 1101                   & 1           & 1           & 1           \\
    \ldots 1110                   & 1           & 1           & 1           \\
    \ldots 1111                   & 1           & 1           & 1           \\
    \bottomrule
  \end{tabular}
\end{table}

\subsection{Prior Work on Synthesis}
\label{ssec:synthesis-related}

  \subparagraph*{Synthesis of Online Algorithms.}
  As mentioned, already the early work on online algorithms considered synthesis in the context of memoryless algorithms~\cite{Chrobak1991}.
  Computer-aided design techniques have also been used to design optimal online algorithms for specific problems. Coppersmith et al.~\cite{Coppersmith1993} studied the design and analysis of randomized online algorithms for $k$-server problems, metrical task systems and a class of graph games; they show that algorithm synthesis is equivalent to the synthesis of random walks on graphs.
  For a variant of the \emph{online knapsack} problem~\cite{Iwama2002}, Horiyama, Iwama and Kawahara~\cite{Takashi2006} obtained an optimal algorithm by using a problem-specific finite automaton and solving a set of inequalities for each of its states. More recently, the synthesis of optimal algorithms for preemptive variants of \emph{online scheduling}~\cite{Baruah1992} was reduced to a two-player graph game~\cite{Chatterjee2018,Pavlogiannis2020}.

  \subparagraph*{Synthesis of Local Algorithms.}
  Synthesis of distributed graph algorithms has a long history, mostly focusing on so-called locally checkable labeling (LCL) problems in the $\local$ model of distributed computing~\cite{balliu2019distributed,balliu2020classification,chang2020distributed,brandt2017lcl,ejc17,rybicki15exact}. In their foundational work, Naor and Stockmeyer~\cite{naor1995can} showed that it is undecidable to determine whether an LCL problem admits a local algorithm in general, but it is decidable for unlabeled directed paths and cycles.
  Balliu et al.~\cite{balliu2019distributed} showed that determining the distributed round complexity of LCL problems on paths and cycles with inputs is decidable, but PSPACE-hard. Moreover, when restricted to LCL problems with binary inputs, there is a simple synthesis procedure~\cite{balliu2020classification}. Recently, Chang et al.~\cite{chang2020distributed} showed that synthesis in unlabeled paths, cycles and rooted trees can be done efficiently.

  Beyond decidability results, synthesis has also been applied in practice to obtain optimal local algorithms.
  Rybicki and Suomela~\cite{rybicki15exact} showed how to synthesize optimal distributed coloring algorithms on directed paths and cycles. Brandt et al.~\cite{brandt2017lcl} gave a technique for synthesizing efficient distributed algorithms in 2-dimensional toroidal grids, but showed that in general determining the complexity of an LCL problem is undecidable in grids. Similarly to our work, Hirvonen et al.~\cite{ejc17} considered the synthesis of \emph{optimization problems}. They gave a method for synthesizing randomized algorithms for the max cut problem in triangle-free \emph{regular graphs}.

  Our work identifies the connection between temporal locality online algorithms and spatial locality in distributed algorithms. As we show in this work, time-local online algorithms can be seen as local graph algorithms on directed paths. However, formally the computational power of the $\local$ model resides between regular and clocked time-local models we study in this work. Thus, decidability results and synthesis techniques do not directly carry over to the time-local online algorithms setting.

\section{Time-Local Algorithms for Online File Migration: Analytical Case Study}\label{sec:file-migration}

We study a variant of online file migration (defined in Section~\ref{sec:regular-algorithms}) in the time-local setting.
Our goals include
\begin{itemize}
  \item deriving analytical lower bounds for competitiveness under limited visible horizon,
  \item showcasing a problem that admits competitive time-local algorithms,
  \item proposing techniques for algorithm design and analysis of time-local algorithms.
\end{itemize}

Note that online file migration problem is not bounded monotone, then we cannot use Theorem~\ref{thm:simulation} to translate classic algorithms to time-local algorithms.
Despite that, we may be able to design competitive algorithms.
The challenge in designing such algorithms lies in the limits of time-local algorithms: not only the past input is unknown, but additionally the algorithm is unaware of its own configuration.

\paragraph*{Known results.}
In the classic full-history online setting, an algorithm Dynamic-Local-Min~\cite{Bienkowski2019} is $4$-competitive.
	The two nodes setting is easier for full-history online algorithms: for networks with 2 nodes, a 3-competitive work function algorithm~\cite{Borodin1992} for metrical task system exists.
In the deterministic online setting, a~lower bound of 3 exists~\cite{Black1989}, by averaging costs of multiple offline algorithms.
In classic online algorithm, randomized algorithms can beat the bound of 3 for deterministic algorithms: for $\costmig \geq 1$, there exists a~$(1+\phi)$-competitive randomized algorithm against the oblivious adversary \cite{Westbrook1994}, where $\phi\approx 1.62$ is the golden ratio.
	For the case of two node networks, the threshold work function algorithm obtains the competitive ratio that approaches $\frac{2e-1}{e-1}\approx 2.581$ as the length of the input sequence grows \cite{Irani1998}.

\subsection{Lower Bounds}
  \label{ssec:file-migration-lb}

In our lower bounds, we leverage the fact that the algorithm is a \emph{function} of the last $T$ requests.
This allows to reason about the performance of a time-local algorithm on different inputs sharing identical subsequences of $T$ requests.

\subparagraph*{A Lower Bound for Time-Local Algorithms.}
We present a lower bound for time-local algorithms for online file migration that shows an inevitable degradation of the competitive ratio when the visible horizon is limited.
The following lower bound assumes the length of the visible horizon is given.
Hardness for simpler settings is implied, in particular for non-clocked time-local algorithms.
Next, we present deterministic algorithms that match this lower bound asymptotically.

\begin{theorem}
  \label{thm:file-migration-lb}
  Fix any randomized $T$-time-local clocked algorithm $\ONL$ for online file migration for networks with at least $2$ nodes.
  Assume that the file size is $\costmig$, and $\costmig \geq T$.
  If~$\ONL$ is $c$-competitive against an oblivious offline adversary, then $c \geq 2\costmig / T$.
\end{theorem}

\begin{proof}
  By Theorem~\ref{thm:random-clock}, mixed clocked time-local algorithms can simulate behavioral clocked time-local algorithms, hence it suffices to consider $\ONL$ as a mixed strategy.

To state $\ONL$'s properties, we consider two infinite sequences of requests, $0^*$ and $1^*$. Later, we will reason about $\ONL$'s performance on finite sequences.
We say that a deterministic algorithm is \emph{resisting} if for each time~$t$ there exists a time~$t' > t$ when the algorithm either outputs $1$ faced with $0^*$, or it outputs $0$ at $t'$ when faced with $1^*$.
A time-local algorithm may be resisting if it has access to a~global clock.

Recall that $\ONL$ is a distribution over deterministic clocked algorithms.
First, assume that $\ONL$ has a resisting strategy in its support.
Consider two input sequence families, $\mathcal{I}_0 := \{0^L : L\in\mathbb{N}\}$ and $\mathcal{I}_1 := \{1^L : L\in\mathbb{N}\}$.
Note that
$\ONL$ may incur an arbitrarily large cost on requests from either of these families of inputs, say $\mathcal{I}_0$, due to an arbitrary number of $0$-requests  served in configuration $1$.
The cost of an optimal offline solution on such sequences is constant:
before serving the sequence, an offline algorithm moves the file to the only node that requests the file.
We conclude that the competitive ratio of $\ONL$ can be arbitrarily large on inputs from $\mathcal{I}_0$ or $\mathcal{I}_1$.

For the rest of this proof, we assume that $\ONL$ does not have any resisting strategy in its support.
A crucial observation is that for a fixed deterministic time-local algorithm, its output on $0^*$ (resp. $1^*$) for any time $\tau$ determines its output on other sequences that contain a sequence of $T$ requests to $0$ (resp. $1$) at time $\tau$.
As $\ONL$ does not have a resisting strategy in its support, after a time~$\tau_{det}$ all strategies in $\ONL$'s support always output $b$ when faced with $T$ many requests to $b$, for $b \in \{0, 1\}$.

Consider an input $\seq = (1^T 0^T)^{L'}$ for some $L'$ to be determined.
Let $\seq'$ be the subsequence of $\seq$ starting from the first $0$-request that comes after $\tau_{det}$.
Fix any optimal offline algorithm $\OFF$ for $\seq$.
  For $\seq \setminus \seq'$, we claim $\ONL(\seq\setminus \seq') \geq \OFF(\seq\setminus \seq')$, where $\ONL(\cdot)$ and $\OFF(\cdot)$ denotes the cost of the algorithm $\ONL$ and an optimal offline algorithm, respectively.
  We analyze $\ONL$ within segments $(1^T 0^T)$ of~$\sigma'$, and we refer to each segment as a \emph{phase}.
  As no strategy in $\ONL$'s support is resisting, and requests from $\seq'$ arrive after $\tau_{det}$, $\ONL$'s behavior on $\seq'$ is determined: it must output $b$ when faced with a $b$-uniform sequence, for $b \in \{0,1\}$.
  Hence, in each phase, $\ONL$ incurs at least the cost $2\costmig$ for changing its output twice.
  On the other hand, $\OFF$ in each phase incurs a cost of at least $T$ --- recall that $T \leq \costmig$, thus either it migrates during the phase and pays $\costmig \geq T$ already, or does not migrate and pays either for all $0$-requests or for all $1$-requests in the phase.
  Summing up the above observations and assuming $\seq'$ consists of~$2T\cdot L'$ requests, we obtain
  \begin{equation*}
    \frac{\ONL(\seq)}{\OFF(\seq)} =
    \frac{\ONL(\seq \setminus \seq')+\ONL(\seq')}{\OFF(\seq\setminus\seq')+\OFF(\seq')} \geq
    \frac{\OFF(\seq\setminus\seq')+2\costmig\cdot L'}{\OFF(\seq\setminus\seq')+T\cdot L'}\enspace.
  \end{equation*}
  By choosing a~long enough sequence $\seq$ (and consequently a large enough $L'$), the competitive ratio can be arbitrarily close to $2\costmig/T$.
\end{proof}
Note that the result presented in this section implies the lower bound for online file migration on general networks (not necessarily consisting of two vertices).

\subparagraph*{A Lower Bound for Small Migration Cost for the Classic Online Model.}
  We investigate the case $\costmig < 1$ in the classic online setting (it is usually assumed that $\costmig \geq 1$). This complements the results obtained in the synthesis of algorithms for file migration from Section~\ref{sec:synthesis}.

\begin{theorem}	\label{th:LB_smallAlpha}
  Consider any deterministic online algorithm \A for online file migration with file size $\costmig$.
  If \A is $c$-competitive, then
	$c \geq 1+1 / \costmig$ for $\costmig \in (0, 1/2]$, and
	$c \geq \min\{ 2+2\costmig, 1+3 / (2\costmig) \}$
	for $\costmig \in (1/2,1)$.
\end{theorem}
\begin{proof}
	Consider an input sequence $\sigma_L$ for any $L\in\mathbb{N}$,
	constructed in the following way.
	We start by issuing $1$-requests until \A migrates the file to node~1.
	Then, we proceed by issuing $0$-requests until \A migrates the file to the node~0.
  We repeat these steps $L$ times.
  Note that \A must eventually perform a~migration, otherwise it is not competitive (an optimal offline algorithm pays at most $2\costmig \cdot L$ for $\sigma_L$).
  In the remainder of the proof, we assume that \A eventually performs a~migration, and consequently $\sigma_L$ is finite.

	We partition $\sigma_L$ into \emph{phases}
	$P_1, \dots, P_L$
	 in the following way.
	The first phase begins with the first request and
	each phase ends when \ALG migrates the file to~0.
  We analyze the ratio of $\A$ to $\OPT$ on each phase separately.
  For any $i \leq L$,
  consider the $i$th phase $P := P_i$.
  Let $x$ be the number of 1-requests and
	$y$ be the number of 0-requests in $P$.
	Then $x,y \geq 1$ and $x+y \geq 2$.
	Recall that \A first serves a~request and then decides
  whether to migrate the file or not.
  Consequently,
  it incurs the cost $2$ in each phase for serving requests remotely, and performs two migrations, and
  its total cost is $x+y+2\costmig \geq 2+2\costmig$.

  Let $\OPT$ be any optimal offline solution.
	Note that \OPT never pays more than $2\costmig$ in any phase:
	it can always migrate the file to~1 prior to
	serving all 1-requests for free,
	and then to~0 prior to serving all 0-requests for free.
	Thus, for any $\costmig > 0$, we have
	\[
	\frac{\ALG(\sigma_L)}{\OPT(\sigma_L)} \geq
	\min_i \frac{\ALG(P_i)}{\OPT(P_i)} \geq
	(2+2\costmig) / 2\costmig = 1 / \costmig + 1.
	\]

	Next, we provide a~stronger bound when $1/2 < \costmig < 1$,
	by distinguishing two cases.
	\begin{description}
	\item[Case 1.] \OPT has the file at the node~0 when it enters the phase.
	If $x=1$ then \OPT does not benefit by migrating the file,
	as otherwise, it would incur for 0-requests in addition to $\costmig$.
	Therefore, it pays $1$ for serving the (single) 1-request remotely,
	and the ratio is $(2+2\costmig) / 1$.
	Else, $x \geq 2 > 2\costmig$, and
	\ALG pays $x+y + 2\costmig \geq 3 + 2\costmig$.
	Since \OPT never pays more than $2\costmig$ for any phase,
	we have $\ALG(P) / \OPT(P) \geq$
	$(3+2\costmig) / 2\costmig = 3 / 2\costmig + 1$.

	\item[Case 2.] \OPT has the file at the node~1 when it enters the phase.
	\OPT serves all 1-requests in the phase for free.
	If~$y=1$, then either \OPT serves the 0-request remotely without migrating the file, paying~$1$,
	or it migrates the file and pays $\costmig < 1$.
	In either case, it pays at most 1 and
	 $\ALG(P) / \OPT(P) \geq (2+2\costmig) / 1$.
	Else, $y \geq 2$, and
	\OPT migrates the file to the node~0 before serving the 0-requests;
	as otherwise it would incur $y \geq 2 > 2\costmig$,
	more than migrating the file twice.
	Since $x+y \geq 3$, we have
	$\ALG(P) / \OPT(P) \geq (3+2\costmig) / \costmig = 3/\costmig+2$.
	\end{description}
\noindent
	Hence, in all cases for $\costmig \in (1/2,1)$,
	we have
	$\ALG(P) / \OPT(P) \geq \min\{ 2+2\costmig, 1+3 / (2\costmig) \}$,
and consequently for all inputs $\sigma_L$ for any $L\in\mathbb{N}$ we have
\[
\frac{\ALG(\sigma_L)}{\OPT(\sigma_L)} \geq
\min_i \frac{\ALG(P_i)}{\OPT(P_i)} \geq
\min\{ 2+2\costmig, 1+3 / (2\costmig) \}
.\]
By combining the results for $\costmig \in (0, 1/2]$ and $\costmig \in (1/2, 1)$, we conclude the lemma.
\end{proof}

A lower bound of 3 is presented in \cite{Black1989}
for $\costmig \geq 1$,
and we note that it holds also for $\costmig<1$.
We summarize all known lower bounds in the following corollary.

\begin{corollary} \label{cor:LB_OFM_det}
	No deterministic classic online algorithm for online file migration can achieve a~competitive ratio less than
	$\max\{ 3, 1+1 / \costmig \}$, for $\costmig > 0$,
	or less than
	$\min\{ 2+2\costmig, 1+3 / (2\costmig) \}$
	for $\costmig \in (0.5,1)$.
\end{corollary}

\subsection{A Deterministic Algorithm for File Migration}
\label{ssec:case-study-det}
\label{ssec:file-migration-det}

We introduce a~constant competitive time-local algorithm,
the \emph{Window~Majority~Algorithm} (\ALG for short),
for the online file migration problem restricted to two nodes,
identified by 0 and 1.

The algorithm takes the last $T$ requests as input,
and it outputs a~value in $\{0,1\}$, the (new) location of the file.
For each request 0 or 1 (the node requesting the file),
\ALG pays a~unit cost if its last output does not match the request
(i.e., if the file is not located at the node).
In such cases, we say \ALG incurs a~\emph{mismatch}.
After serving the request,
\ALG may choose to migrate the file to the other node (by switching the output) at the cost $\costmig$, and we say \ALG \emph{flips} its output.

The algorithm scans the visible horizon looking for a~distinguished subsequence
of requests,
called a~\emph{relevant window}.
It decides its output based on 1) the existence of any relevant window
and 2) an invariable property of the relevant window (if any exists).
After a~window enters the visible horizon,
\ALG maintains its latest output as long as
(i) the window is contained in the visible horizon
(while \emph{sliding}), and
(ii) it is not succeeded by a~more recent relevant window.
\ALG may flip its output once (i) or (ii) is no longer the case.
Intuitively,
a~relevant window serves as a~short-living memory,
enabling \ALG to maintain the same output for as long as the visible horizon contains the window.

For $T \geq 6$, $\lambda = \min \{\lceil T/6 \rceil , \costmig \} \geq 1$ and $b \in \{0,1\}$,
a~\emph{$b$-window} is a~subsequence of length~$3\lambda$,
 in which the number of $b$-requests is at least twice the number of $\bar{b}$-requests, $\bar{b} = 1-b$.
\ALG outputs 1 only if the most recent $b$-window
in the visible horizon is a~1-window.
Hence,
it outputs 1 as long as the visible horizon contains a~1-window
 that is not succeeded by a~(more recent) 0-window.
The 1-window slides further to the past as new requests arrive
until it is no longer contained in the visible horizon.
At this moment,
\ALG either flips back to 0 (as the default output), or it maintains the output~1
because of a~more recent 1-window in the visible horizon.

The algorithm \ALG takes the visible horizon~$T$ (i.e., the past~$T$ requests) as a parameter and
outputs~0 or~1 according to the following rules.
\begin{description}
\item[Rule~1.]
	Output $b \in \{0,1\}$ if the most recent window in the visible horizon is a~$b$-window.
\item[Rule~2.]
	Output 0 if the visible horizon contains no $b$-window for $b \in \{ 0, 1\}$.
\end{description}

\noindent
Note that \ALG flips to~0 either because the visible horizon contains a~$0$-window (Rule~1) or there is no $b$-window in the visible horizon (Rule~2).

\subsubsection{Analysis of the Deterministic Algorithm for File Migration}
We show the following theorem about \ALG:
\begin{theorem}
  \label{cor:6-competitive}
  \label{thm:6-competitive}
	The time-local algorithm \ALG is $c$-competitive, where  $c=6$ for $T \geq 6\costmig$,
	$c = 4+\frac{12\costmig}{T}$  for $6 \leq T \leq 6\costmig$, and
	$c = 4+ 2\costmig$ for $1 \leq T \leq 6$.
\end{theorem}
Our analysis follows a sliding window method, and classifies parts of the input adequately to (1) the changing majority and (2) a lower bound for offline cost.
We leverage the fact that
\ALG neither flips too frequently,
incurring excessive reconfiguration cost,
nor too conservatively,
incurring excessive mismatches.
We focus our attention on individual subsequences between two consecutive flips to~1.
We show that any of these subsequences contains sufficiently many
0 and 1-request, so that
an optimal offline algorithm incurs a~cost within  a~factor $O(\costmig/T)$ of \ALG's cost.

We begin with auxiliary definitions and notations
that we use in our analysis.
A \emph{(sub)segment} of $\seq$ between requests $\sigma_i$ and $\sigma_j, j > i$,
is a~contiguous subsequence of $\seq$ denoted by $\seq(i,j]$.
We denote the concatenation of any two consecutive
segments $\S_1$ and $\S_2$ by $\S_1\S_2$.
We denote the $b$-window that is contained in the visible horizon 
starting from the time~$\tau$ by $\W_\tau := \seq(\tau- 3\lambda, \tau]$.
For $b \in \{0,1\}$,
we denote the number of $b$-requests in a~segment~$\S$ by $n_{b} (\S)$.
We say that \ALG flips \emph{at (the time)~$\tau$} if \ALG flips immediately after the request~$\sigma_\tau$.

\textbf{Cost Notations.}
Assume \ALG flips to~$b$ at~$\tau$.
If the flip occurs by Rule~1 then
	$n_b(\W_{\tau}) = 2\lambda$ and~$n_{\bar{b}}(\W_{\tau}) = \lambda$.
Else, it occurs by Rule~2 and $n_b(\W_{\tau}), n_{\bar{b}}(\W_{\tau}) \leq 2\lambda-1$,
as otherwise the flip at~$\tau$ would occurs by Rule~1.
We compare the cost of our algorithm to the cost of 
an optimal offline algorithm denoted by $\OPT$.
The total cost incurred by \ALG and \OPT
while serving a~segment $\S$ is denoted by
$\ALG(\S)$ and $\OPT(\S)$ respectively.
We denote the cost of mismatches to \ALG for a~segment $\S$
by $\mis(\S)$.
Therefore,
$\ALG(\W_\tau) = \mis(\W_\tau) + \costmig \leq 2\lambda+\costmig$.
We present additional properties of \ALG that we use for the analysis of its competitive ratio.

\begin{lemma}	\label{lemma:BlockSegment}
	If \ALG does not flip in a~segment~$\S$ where $|\S| = 3\lambda$ then
	\[
	\OPT(\S) \geq \min \{ n_{0}(\S), n_{1}(\S), \costmig \} \geq \mis(\S) / 2
	.\]
\end{lemma}
\begin{proof}
	Assume \ALG outputs $b \in \{0,1\}$ in $\S$.
	Then it incurs the cost of mismatches to~$\bar{b}$-requests, that is,
	$\mis(\S) = n_{\bar{b}}(\S)$.
	Moreover, we have $n_{\bar{b}}(\S) \leq 2\lambda-1$,
	as otherwise \ALG would flip to~$\bar{b}$ in this segment.
	Therefore,
	$n_b(\S) = 3\lambda-n_{\bar{b}}(\S) \geq \lambda+1 > n_{\bar{b}}(\S)/2 = \mis(\S) / 2$.
	
	\OPT pays mismatches to~$b$-requests or to~$\bar{b}$-requests,
	or it performs a~flip and possibly serves some of them for free.
	Then regardless of \OPT's actions in $\S$, it incurs
	\[
	\OPT(\S) \geq \min \{ n_b(\S), n_{\bar{b}}(\S), \costmig \}
	\geq \min \{ \mis(\S)/2, \mis(\S), \costmig \} =
	\min \{ \mis(\S)/2, \costmig \}
	,\]
	and the claim follows since $\mis(\S) \leq 2\lambda-1$ and thus
	$\mis(\S)/2 < \lambda \leq \costmig$.
\end{proof}

Assume \ALG flips to~$b$ at the end of a~segment~$\S$ where $|\S|=3\lambda$.
Then $\S$ contains at most~$2\lambda$ many $b$-requests,
 as otherwise the flip to~$b$ would occur earlier.
Moreover, $\S$ contains at most $2\lambda-1$ many $\bar{b}$-requests,
otherwise \ALG would output $\bar{b}$ at the end of $\S$ (by Rule~1),
contradicting our assumption.
Therefore,
$\S$ contains at least $\lambda$ many $b$- and $\bar{b}$-requests,
that is,
$n_{b}(\S), n_{\bar{b}}(\S) \geq \lambda$.
Using Lemma~\ref{lemma:BlockSegment},
we conclude this observation as follow.
\begin{corollary}	\label{cor:OPT(W)}
	If \ALG flips to~$b$ (by either rules) at the end of a~segment~$\S$, $|\S|=3\lambda$,
	then $\OPT(\S) \geq \lambda.$
\end{corollary}

\subparagraph*{Phase Analysis.}
We subdivide the input sequence into consecutive segments such that \ALG flips only at the end of each segment.
Specifically, we partition the input as
\begin{equation}	\label{eq:partitioning}
	\seq = \R_\first \L_1 \R_1 \dots \L_m \R_m \S_\last,
\end{equation}
where each segment $\L_*$ ends with a~flip to~0 and each segment $\R_*$ ends with a~flip to~1, 
and  the segment $\S_\last$ is the remainder after the last flip to~1 (at the end of $\R_m$).
Fix any segment
$\F \in \{\R_\first\} \cup \{\L_i,\R_i\}_i$.
Let $b \in \{0,1\}$ be the value to which \ALG flips at the end of $\F$.
That is, \ALG outputs $\bar{b}$ for the entire segment and 
outputs~$b$ immediately after the last request.

\begin{lemma}	\label{lemma:Rule2|P_i|}
 If the flip at the end of $\F$
  occurs by Rule~2 then $|\F| > 3\lambda$.
\end{lemma}
\begin{proof}
	By \ALG's definition, the flip by Rule~2 is a~flip to~0.
	Let~$\tau$ denote the time when this flip occurs.
	Therefore
	\ALG flips to~1 the end of the segment preceding $\F$ in~(\ref{eq:partitioning}) 
	at a~time denoted by~$\tau' < \tau$, that is, immediately after the 1-request $\sigma_{\tau'}$.
	By definition, $T \geq 6\lambda$, and therefore
	the segment $\W_{\tau'}$ is contained in the visible horizon
	after each request $\sigma_{\tau'},\sigma_{\tau'+1}, \dots, \sigma_{\tau' + 3\lambda}$.
	Thus,
	\ALG outputs 1 at every time step
	$\tau', \dots ,\tau'+3\lambda$.
	Hence,
	the flip to~0 by Rule~2 does not occur earlier than $\tau'+3\lambda+1$,
	and therefore $|\F| \geq 3\lambda+1$.
\end{proof}

In the next two lemmas, we derive more fine-grained lower bounds for the cost incurred by~\OPT
when the inter-flip segment~$\F$ is sufficiently long, that is, when $|\F| \geq 3\lambda$.
Assume \ALG flips to~$b \in \{0,1\}$ at the end of this segment.
Consider the partitioning $\F = \U \V \W$,
where $|\U|$ is a~multiple of~$3\lambda$,
$|\V|<3\lambda$ and $|\W|=3\lambda$.
Note that since \ALG does not flip in $\V$, it pays
$\ALG(\V) = \mis(\V) \leq 2\lambda-1$.
\ALG flips at the end of $\W$ and by Corollary~\ref{cor:OPT(W)}, it pays
$\ALG(\W)=\mis(\W)+\costmig \leq 2\lambda+\costmig$.
Using this observation, we conclude ALG's mismatch cost in the following statement.
\begin{observation} \label{obs:ALG(VW)}
	If $|F| \geq 3\lambda$ then
	$\mis(\V\W) = \mis(\V) + \mis(\W) \leq 2\lambda-1+2\lambda < 4\lambda$.
\end{observation}

We provide competitive ratios for each subsegment of $\F$ independently.
We subdivide the segment $\U$ into subsegments of length $3\lambda$.
By applying Lemma~\ref{lemma:BlockSegment} to each one separately,
we conclude the ratio for $\U$ in the following statement.

\begin{observation}	\label{obs:OPT(U)}
For the segment $\U$, it holds that $\OPT(\U) \geq \ALG(\U)/2$.
\end{observation}

The next lemma lower-bounds costs to \OPT
for a~segment in which \OPT does not flip and
\ALG flips at the end of the segment.

\begin{lemma} \label{OPTNoFlip}
	Assume \ALG serves $\F$ while outputting $\bar{b}$ and  flips to~$b$ at the end of this segment.
	If \OPT does not flip in $\V\W$ then
	one of the two cases holds:
	\begin{enumerate}[i)]
		\item	\label{OPTState=notb}
		\OPT serves $\V\W$ in state $\bar{b}$ and
		$\OPT(\V\W) = \mis(\V\W) \geq \mis(\V)+\lambda$
		\item	\label{OPTState=b}
		\OPT serves $\V\W$ in state $b$ and
		$\OPT(\V\W) \geq \lambda$.
		
	\end{enumerate}
	
\end{lemma}
\begin{proof}
	Regardless of \OPT's state,
	for $\W$, it incurs
	$\OPT(\W) \geq  \lambda$ by Corollary~\ref{cor:OPT(W)}.	
	If \OPT serves the entire $\V\W$ in state $\bar{b}$ then
	both algorithms \OPT and \ALG pay mismatches to~$\bar{b}$'s
	in this segment and $\OPT(\V\W) = \mis(\V\W)$,
	which concludes Lemma~\myref{OPTNoFlip}{OPTState=notb}.
	Otherwise,
	\OPT serves the entire $\V\W$ in state $b$ and
	possibly pays no cost for $\V$ and at least $\lambda$ for $\W$,
	which concludes Lemma~\myref{OPTNoFlip}{OPTState=b}.
\end{proof}

\noindent
The next lemma lower-bounds \OPT's cost for $\V\W$ when
\OPT does flip in this segment.

\begin{lemma} \label{OPTFlips}
	Assume \ALG flips to~$b$ at the end of~$\F$.	
	If \OPT flips in $\V\W$
	then one of the two cases applies:
	\begin{enumerate}[i)]
		\item \label{flipTob}
		\OPT flips to~$b$ in $\V\W$ and
		$\OPT(\V\W) \geq \min\{\mis(\V), \lambda \} + \costmig$,
		
		\item \label{flipToNotb}
		\OPT flips to~$\bar{b}$ in $\V\W$ and
		$\OPT(\V\W) \geq \max\{\mis(\V\W)/2-\lambda, 0\} + \costmig$.
		
	\end{enumerate}
	
\end{lemma}
\begin{proof}
	If \OPT flips more than once in $\V\W$ then
	$\OPT(\V\W) \geq 2\costmig$.
	Since $\lambda \leq \costmig$, we have
	$\OPT(\V\W) \geq 2\costmig \geq \min\{\mis(\V),\lambda \} + \costmig$
	concluding Lemma~\myref{OPTFlips}{flipTob}.
	Since
	$\mis(\V\W) < 4\lambda$, we have
	$\mis(\V\W)/2 -\lambda < \lambda \leq \costmig$.
	Thus, if \OPT flips more than once in $\V\W$ then
	$\OPT(\V\W) \geq 2\costmig \geq \max\{\mis(\V\W)-2\lambda, 0\}  + \costmig$ and Lemma~\myref{OPTFlips}{flipToNotb} holds.
	Hence, in the remainder, we assume \OPT flips only once in $\V\W$.

	\begin{itemize}
	\item
		\textbf{\OPT flips to $\bm{b}$ in $\V\W$.}
		If \OPT flips in $\V$ then it does not flip in $\W$ and
		by Corollary~\ref{cor:OPT(W)},
		$\OPT(\W) \geq \lambda$.
		Hence,
		$\OPT(\V\W) \geq \costmig+\lambda
		\geq \costmig + \min\{\mis(\V), \lambda \}$.
		Otherwise,
		\OPT serves~$\V$ in the state $\bar{b}=s$ and it flips to~$b$ in $\W$ paying $\costmig$.
		Then $\OPT(\V) = n_b(\V) =\mis(\V)$, and
		Lemma~\myref{OPTFlips}{flipTob} follows from
		$\OPT(\V\W) \geq \mis(\V) + \costmig
		\geq \min\{\mis(\V), \lambda \} + \costmig
		.$
		
		\item
		\textbf{\OPT flips to $\bm{\bar{b}}$ in \bm{$\V\W$}.}
		If \OPT flips in $\V$ then it serves $\W$ without flipping and by Corollary~\ref{cor:OPT(W)} it pays at least $\lambda$ for $\W$.
		Using $\mis(\V\W) < 4\lambda$ (Observation~\ref{obs:ALG(VW)}), we obtain
		\[
		\OPT(\V\W) \geq \costmig + \lambda
		\geq \costmig + \max\{\mis(\V\W)/2 - \lambda, 0\}
		.\]
		Otherwise,
		\OPT serves the entire $\V$ in state~$b$ (possibly for free) and flips to~$\bar{b}$ in $\W$.
		We lower-bound the cost of mismatches in $\W$ after \OPT flips to~$\bar{b}$ as follows.
		If $n_b(\V) \geq 1$
		then there must exist at least $\lambda+1$ many $\bar{b}$-requests
		between the last $b$-request in $\V$ and \ALG's flip to~$b$ at~$\tau$.
		Assume this is not the case and $n_{\bar{b}}(\W) \leq \lambda$.
		Therefore $n_b(\W) = 3\lambda - n_{\bar{b}}(\W) \geq 2\lambda$ and
		$ n_b(\V\W) \geq 2\lambda+1$.
		This implies that \ALG flips to~$b$ earlier than $\tau$ in $\W$,
		contradicting our assumption.
		Therefore, at least $\lambda+1$ many $\bar{b}$'s occur in $\V\W$
		between the last $b$-request in $\V$ and \ALG's flip to~$b$.
		
		Next, we lower-bound the number of mismatches incurred by \OPT (in $\W_{\tau}$) after it flips to~$\bar{b}$.
		Let $\sigma_p=1$ be the $(\lambda+1)$th $\bar{b}$-request in $\V\W$
		(which exists as shown).
		Either \OPT flips to $\bar{b}$ after $\sigma_p$,
		that is, after paying at least $\lambda+1$ mismatches to the $\lambda+1$ many $\bar{b}$-requests in $\V\W$,
		or it flips to $\bar{b}$ before serving $\sigma_p$.
		In the latter case,
		\OPT pays mismatches to the remaining $b$-requests
		(occurring after $\sigma_p$) in $\W$.
		Let $x$ be the number of $b$-requests in $\V\W$
		after $\sigma_p$.
		Then the number of $b$-requests in $\V\W$
		and before $\sigma_p$
		is $n_b(\V\W) - x < 2\lambda$,
		otherwise \ALG would flip to~$b$ earlier than $p < \tau$,
		contradicting our assumption,
		which implies
		$x > n_b(\V\W) - 2\lambda = \mis(\V\W) - 2\lambda
		\geq \max\{\mis(\V\W)/2 - \lambda, 0\}$.
  \end{itemize}
	Therefore in any case,
	\OPT pays at least $\mis(\V\W)/2 - \lambda$
	mismatches to~$b$-requests after it flips to~$\bar{b}$,
	which concludes Lemma~\myref{OPTFlips}{flipToNotb}.
\end{proof}

Recall that when the flip to~$b$ at~$\tau$ is triggered by Rule~1,
the window $\W_\tau$ consists of~$2\lambda$ many $b$'s and $\lambda$ many $\bar{b}$'s.
If $|\F| < 3\lambda$, i.e., it is \emph{short}, then $\F$ may contain only a~subset of
these $b$-requests.
That is, $n_b(\F) \leq n_b(W_\tau) = 2\lambda$.
The next lemma states that this subset cannot be too small,
which later is used to deduce a~significant cost for $\OPT$.
\begin{lemma} \label{lemma:atLeastLambda}
	Assume the flip (to~$b$) at the end of the $\F$ occurs by Rule~1.
	If also the preceding flip (to~$\bar{b}$) occurs by Rule~1 then
	$n_b(\F) \geq \lambda$.
\end{lemma}
\begin{proof}
	Let $\tau$ and $\tau' < \tau$ denote times (i.e., request indices), respectively,
	at the end of $\F$ and at the end of the segment preceding $\F$ in~(\ref{eq:partitioning}).	 
	Assume for contradiction that $x:=n_b(\F) < \lambda$.
	Since the flip at~$\tau$ occurs by Rule~1,
	$\W_{\tau}$ must contain exactly $2\lambda$ many $b$-requests.
	Then the remaining $2\lambda - x > \lambda$  of these $b$-requests
	must be in~$\W_{\tau'}$, that is,
	$n_{b}(\W_{\tau'}) > \lambda$.
	Hence, the number of $b$-requests in $\W_{\tau}$
	is $3\lambda - n_b(\W_{\tau'}) < 2\lambda$,
	contradicting our assumption that the flip at~$\tau'$ occurs by Rule~1.
\end{proof}
%

We refer to the segment $\P_i = \L_i\R_i$ as a~\emph{phase}.
We analyze the total cost to \OPT and the competitive ratio for each phase independently.
Theorem~\ref{thm:ALG} aggregates these individual ratios into one competitive ratio for~$\seq$.
 
Fix any phase $\P_i$.
The last flip prior to $\L_i$ is a~flip to~1 which always occurs by Rule~1.
Then Lemma~\ref{lemma:atLeastLambda} implies $n_0(\L) \geq \lambda$.
Next, we conclude the ratio within a~phase.

\begin{lemma} \label{lemma:phaseRatio}
  For any phase $\P_i := \L_i\R_i$, we have
  $\ALG(\P_i) / \OPT(\P_i) \leq 4 + \frac{2\costmig}{\lambda}$.
\end{lemma}
\begin{proof}
Let~$l$ denote the time \ALG flips to~0 at the end of $\L :=\L_i$ and let $r$ denote the time it flips to~1 at the end of $\R:=\R_i$.
By these definitions,
$\ALG(\P_i) = \ALG(\L) + \ALG(\R)$ and
$\OPT(\P_i) = \OPT(\L) + \OPT(\R)$.
If $\L$ is long, that is $|\L| \geq 3\lambda$,
we partition it as $\L = \U_l \V_l \W_l$ where
$|\U_l|$ is a~multiple of $3\lambda$, $|\V_l| < 3\lambda$,
and $|\W_l| = 3\lambda$.
We partition $\R$ as $\R = \U_r \V_r \W_r$  in a~similar way whenever $|\R| \geq 3\lambda$.

We bound the costs under four major cases of $|\L|$ and $|\R|$,
as depicted in Figure~\ref{fig:phase}.

\begin{itemize}
  \item 
\textbf{Both parts are short.}
As already shown, it always holds that $n_0(\L) \geq \lambda$.
Since $|\L| < 3\lambda$,
Lemma~\ref{lemma:Rule2|P_i|} implies that the flip at~$l$ occurs by Rule~1.
Then,
Lemma~\ref{lemma:atLeastLambda} guarantees
$n_1(\R) \geq \lambda$.
Hence,
\OPT (regardless of its state) either does not flip and incurs at least $\lambda$ mismatches in this phase or it flips (paying $\costmig$ and possibly less mismatches).
In any case,
$\OPT(\P_i) \geq \min\{\lambda,\costmig\} = \lambda$.
Since $\mis(\W_l), \mis(\W_r) \leq 2\lambda + \costmig$,
we have
\begin{align}
  \frac{\ALG(\P_i)}{\OPT(\P_i)} =
  \frac{\ALG(\L) + \ALG(\R)}{\OPT(\P_i)}
  \leq \frac{\mis(\W_l) + \mis(\W_r) + 2\costmig}{\OPT(\P_i)}
  \leq \frac{4\lambda + 2\costmig}{\lambda}
  \label{UB:shortshort}
  .
\end{align}

\begin{figure}[t]
	\centering
	\begin{minipage}{.73\textwidth}
		\includegraphics[scale=1]{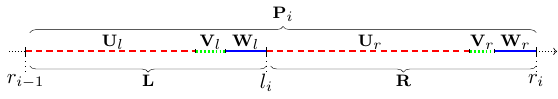}
		\subcaption{} \label{fig:longlong}
	\end{minipage}%
	\begin{minipage}{.27\textwidth}
	\includegraphics[scale=1]{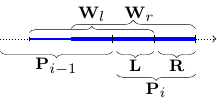}
	\subcaption{} \label{fig:shortshort}
\end{minipage}

	\begin{minipage}{.5\textwidth}
		\includegraphics[scale=1]{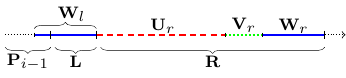}
		\subcaption{} \label{fig:shortLong}
	\end{minipage}%
	\begin{minipage}{.5\textwidth}
		\includegraphics[scale=1]{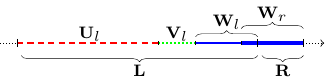}
		\subcaption{} \label{fig:longShort}
	\end{minipage}

	\caption{
		\ALG flips to 0 at the end of segment $\L$ and flips to 1 at the end of segment $\R$.
		In \ref{fig:longlong},
		both parts $\L$ and $\R$ are long, implying that both windows~$\W_l$ and~$\W_r$ are (fully) contained in their respective parts and therefore in the phase~$\P_i=\L\R$.
		In~\ref{fig:shortshort}, both windows are short, 
		implying that possibly both windows are only partially contained in~$\P_i$, hence overlapping the preceding phase~$\P_{i-1}$.
		In \ref{fig:shortLong},
		the left part is short, and only the left window may overlap the preceding phase.
		In~\ref{fig:longShort}, the left part is long, implying both windows are (fully) contained in~$\P_i$.
	}
	\label{fig:phase}
\end{figure}

\item \textbf{Both parts are long.}
See Figure~\ref{fig:longlong} for an illustration.
Regardless of \OPT's actions in~$\L$,
we have
$\OPT(\U_l) \geq \mis(\U_l)/2$ from Observation~\ref{obs:OPT(U)}, and
$\OPT(\W_t) \geq \lambda$ from Corollary~\ref{cor:OPT(W)}.
Therefore,
$\OPT(\L) = \OPT(\U_l) + \OPT(\V_l\W_l)
  \geq \mis(\U_l) / 2 + \lambda$.
Using a~similar inequality for $\R$,
we obtain
\[
\OPT(\P_i) = \OPT(\L)+\OPT(\R) \geq \mis(\U_l) / 2 + \mis(\U_r) / 2+ 2\lambda
.\]
By (Observation~\ref{obs:ALG(VW)}),
\ALG's mismatch cost in the remainder of $\L$ is
	$\mis(\V_l\W_l)  < 4\lambda$,
 and similarly $\mis(\V_r\W_r)  < 4\lambda$.
Then,
\begin{align*}
	\ALG(\P_i) &= \mis(\U_l) + \mis(\V_l\W_l) + \costmig
  + \mis(\U_r)  + \mis(\V_r\W_r)  + \costmig	\\
  &\leq \mis(\U_l) + \mis(\U_r)  + 8\lambda + 2\costmig
  ,~\text{and}
\end{align*}
\begin{align}
  \frac{\ALG(\P_i)}{\OPT(\P_i)}
  \leq
  \frac
  {\mis(\U_l)  + \mis(\U_r)  + 8\lambda + 2\costmig}
  {\mis(\U_l)/2 + \mis(\U_r)/2 + 2\lambda}
  \leq \frac{8\lambda + 2\costmig}{2\lambda}
  =	4+ \frac{\costmig}{\lambda}
  \label{UB:longlong}
  .
\end{align}

\item
\textbf{Only the right part is long.}
Then, $\L$ is a~subsegment of $\W_l$
(see Figure~\ref{fig:shortLong}) and
therefore $\ALG(\L) \leq \ALG(\W) \leq 2\lambda + \costmig$.
Thus, we have,
\begin{align*}
	\ALG(\P_i) &= \ALG(\L) + \ALG(\U\V\W) \leq 
	(2\lambda+\costmig)+ \mis(\U_r) + \mis(\V_r) + 2\lambda + \costmig \\
	&\leq \mis(\U_r) + \mis(\V_r) + 4\lambda + 2\costmig
	.
\end{align*}
We distinguish several cases for its cost in~$\R$.

\begin{itemize}
\item
\textbf{Case~1.1.}
\OPT enters $\V_r\W_r$ in state 0 and serves it in state 0.
In this case,
Lemma~\myref{OPTNoFlip}{OPTState=notb} applies
which together with Observation~\ref{obs:OPT(U)} yields
\[
\OPT(\R) = \OPT(\U_r) + \mis(\V_r) + \mis(\W_r)
 \geq \mis(\U_r)/2 + \mis(\V_r) + \lambda
  ,\]
and thereby
\begin{align}
  \frac{\ALG(\P_i)}{\OPT(\P_i)}
  \leq	\frac
  { \mis(\U_r) + \mis(\V_r) + 4\lambda + 2\costmig}
  { \mis(\U_r)/2 + \mis(\V_r) + \lambda}
  \leq
  \frac{4\lambda + 2\costmig }{\lambda}
  =  4 + \frac{2\costmig}{\lambda}
  \label{UB:shortlongIn0}
  .
\end{align}

\item
\textbf{Case~1.2.}
\OPT enters $\V_r\W_r$ in state 0 and flips to~1 in $\V_r\W_r$.
In this case,
Lemma~\myref{OPTFlips}{flipTob} applies which
together with Observation~\ref{obs:OPT(U)} yields
\[
\OPT(R) = \OPT(\U_r) + \OPT(\V_r\W_r)
  \geq \mis(\U_2)/2 + \min\{ \mis(\V_r), \lambda\} + \costmig
  .\]
By distinguishing the two cases
$\mis(\V_r) < \lambda$ and $\mis(\V_r) \geq \lambda$,
we obtain
\begin{equation}
\label{UB:shortlong_fliptTo1}
\begin{split}
  \frac{\ALG(\P_i)}{\OPT(\P_i)} &\leq
  \frac
  { \mis(\U_r) + \mis(\V_r) + 4\lambda + 2\costmig}
  { \mis(\U_r)/2 + \min\{ \mis(\V_r),  \lambda\} + \costmig} 
  \leq \max\{
  \frac{4\lambda + 2\costmig }{\costmig},
  \frac{6\lambda + 2\costmig }{\lambda + \costmig}
  \} \\
  &\leq  4 + \frac{2\costmig}{\lambda}
  .
\end{split}
\end{equation}

\item
\textbf{Case~1.3.}
\OPT enters $\V_r\W_r$ in state 1.
Recall that $n_0(\L) \geq \lambda$.
Either \OPT serves the entire $\L$ in state 1 and
$\OPT(\L) = n_0(\L) \geq \lambda$,
or it flips in $\L$ at cost $\costmig \geq \lambda$.

If \OPT serves $\V_r\W_r$ in state 1 then by
Corollary~\ref{cor:OPT(W)},
$\OPT(\V_r\W_r) \geq \OPT(\W_r) \geq \lambda$.
Else,
it flips to~0 and $\OPT(\V_r\W_r) \geq \costmig \geq \lambda$.
Therefore in any case of \OPT's action, we have
$\OPT(\V_r\W_r) \geq \lambda$,
and by applying Observation~\ref{obs:OPT(U)} to~$\U_r$,
we obtain
\[
\OPT(\P_i) \geq \OPT(\L) + \OPT(\U_r) + \OPT(\V_r\W_r)
  \geq \lambda + \mis(\U_r)/2 + \lambda
  .\]
Since $\ALG(\V_r\W_r) < 4\lambda+\costmig$, we have
\begin{align}
  \frac{\ALG(\P_i)} {\OPT(\P_i)}
  \leq
  \frac
  {\mis(\U_r)+ 6\lambda + 2\costmig}
  {\mis(\U_r)/2 + 2\lambda}
  \leq
  \frac{6\lambda + 2\costmig }{2\lambda}
  = 3 + \frac{\costmig}{\lambda}
  \label{UB:shortlong_entersVWIn1}
  .
\end{align}
\end{itemize}
\item
\textbf{Only the left part is long.}
See Figure~\ref{fig:longShort} for an illustration.
We distinguish cases of \OPT's state when in enters $\V_l\W_l$.
Note that \ALG's flip to~0 at $l$ possibly occurs by Rule~2.

\begin{itemize}
\item
\textbf{Case~2.1.}
\OPT enters $\V_l\W_l$ in state 1 and serves it in state 1.
This case is symmetric to Case 1.1 and
the upper bound (\ref{UB:shortlongIn0}) holds analogously,
after swapping the usage of $\R$ and $\L$, as well as $0$'s and $1$'s.

\item
\textbf{Case~2.2.}
\OPT enters $\V_r\W_r$ in state 1 and flips to~0 later in this segment.
This case is symmetric to Case 1.2 and
the upper bound (\ref{UB:shortlong_fliptTo1}) holds analogously,
after swapping the usage of $\R$ and $\L$, as well as $0$'s and $1$'s.

\item
\textbf{Case~2.3.}
\OPT enters $\V_l\W_l$ in state 0.
Hence, serving 1-requests is costly for \OPT unless it flips.
First, we lower bound the number of these 1-requests.
Recall that if the flip to~0 is by Rule~2 then
possibly~$n_1(\R) < \lambda$ and \OPT may incur an insignificant cost for~$\R$.
However, \OPT incurs a~significant cost for~$\W_l\R$ as we show next.

Since $|\L| \geq 3\lambda$,
the window $\W_r$ is contained in the segment $\W_l\R$.
Since the flip at the end of $\R$ is by Rule~1,
$n_1(\W_l\R) \geq n_1(\W_r) = 2\lambda$.
Either \OPT serves the entire segment $\V_l\W_l\R$ in state 0 and incurs
$\OPT(\V_l\W_l\R) \geq n_1(\W_l\R)  \geq 2\lambda $ mismatches,
or it flips to~1 in this segment.
If \OPT flips to~1 in $\R$ then from Corollary~\ref{cor:OPT(W)} and (Observation~\ref{obs:ALG(VW)}), 
we obtain
\[
\OPT(\V_l\W_l\R) \geq \OPT(\W_l) +  \OPT(\R) \geq \lambda + \costmig  \geq 2\lambda
	\geq \max\{\mis(\V_l\W_l)/2-\lambda, 0\} + \lambda
.\]
Else, \OPT flips to~1 in $\V_l\W_l$,
and by applying Lemma~\myref{OPTFlips}{flipToNotb} to the segment $\L$, we obtain
\begin{equation}	\label{eq:OPT(V_lW_l)}
	\OPT(\V_l\W_l) \geq \max\{\mis(\V_l\W_l)/2-\lambda, 0\} + \lambda
	.
\end{equation}
Thus, (\ref{eq:OPT(V_lW_l)}) holds in any case where \OPT enters $\V_l\W_l$ in state~0.
By applying Observation~\ref{obs:OPT(U)} to $\U_l$,
we obtain
\begin{equation} \label{eq:entersVWIn0}
\OPT(\P_i) \geq \OPT(\U_l) + \OPT(\V_l\W_l\R) \geq \mis(\U_l)/2
  + \max\{\mis(\V_l\W_l)/2- \lambda, 0\} + \lambda.
\end{equation}
If $\mis(\V_l\W_l) < 2\lambda$
	then (\ref{eq:entersVWIn0}) reduces to
$\OPT(\P_i) \geq \mis(\U_r)/2  + \lambda$.
Using
\[\ALG(\P_i)
  \leq \mis(\U_l)+\mis(\V_l\W_l)+ \costmig + (2\lambda + \costmig)
  \leq\mis(\U_l) + 4\lambda + 2\costmig,
  ~\text{and}
  \]
\begin{align}
  \frac{\ALG(\P_i)} {\OPT(\P_i)}
  \leq
  \frac
  {\mis(\U_r)+ 4\lambda + 2\costmig}
  {\mis(\U_r)/2 + \lambda}
  \leq
  \frac{4\lambda + 2\costmig }{\lambda}
  = 4 + \frac{2\costmig}{\lambda}
  \label{UB:longshort_entersVWIn0_A}
  .
\end{align}
Else,
$\mis(\V_l\W_l) \geq 2\lambda$ holds and (\ref{eq:entersVWIn0}) reduces to
$\OPT(\V_l\W_l) \geq (\mis(\V_l\W_l)/2 - \lambda)+\lambda$.
Let~$z := \mis(\V_l\W_l)/2- \lambda$.
Then,
$ \mis(\V_l\W_ l) = 2z+2\lambda$, and
\[
\ALG(\P_i)
  \leq \mis(\U_l)+\mis(\V_l\W_l)+ 2\lambda + 2\costmig
  =  \mis(\U_l) + (2z+2\lambda) + 2\lambda + 2\costmig
  ,~\text{and}
\]
\begin{align}
  \frac{\ALG(\P_i)} {\OPT(\P_i)}
  \leq
  \frac
  {\mis(\U_r)+ 2z + 4\lambda + 2\costmig}
  {\mis(\U_r)/2 + z + \lambda}
  \leq
  \frac{4\lambda + 2\costmig }{\lambda}
  = 4 + \frac{2\costmig}{\lambda}
  \label{UB:longshort_entersVWIn0_B}
  .
\end{align}
\end{itemize}
\end{itemize}
From all upper bounds (\ref{UB:shortshort})--(\ref{UB:longshort_entersVWIn0_B}),
we conclude
$\ALG(\P_i) / \OPT(\P_i)
  \leq 4 + \frac{2\costmig}{\lambda}$.
\end{proof}

\begin{theorem} \label{thm:ALG}
  For any input sequence $\seq$, any $T \geq 6$ and $1 \leq  \lambda \leq \costmig$,
  we have
  \[{
  	\ALG(\seq) \leq
        \Bigl(4+ \frac{2\costmig}{\lambda}\Bigr) \OPT(\seq) + 6\costmig}
    .\]
\end{theorem}
\begin{proof}
Assume \ALG performs at least one flip in $\seq$.
Consider the partitioning~(\ref{eq:partitioning}).
From  Lemma~\ref{lemma:phaseRatio},
we have
$\ALG(\L_i\R_i) \leq (4+2\costmig / \lambda) \OPT(\L_i\R_i)$.

In the remainder,
we upper bound the ratio separately for $\R_\first$ and $\S_\last$.
Recall that \ALG starts serving $\seq$ by outputting~0
until it flips to~1  at for the first time at the end of~$\R_\first$.
We distinguish two cases for this segment.

\begin{itemize}
\item
\textbf{$\R_\first$  is short.}
In this case, $\ALG(\R_\first) \leq 2\lambda+\costmig$.
\OPT begins in state 0 and
either  pays~$2\lambda$ mismatches to $1$-requests in this segment,
or it flips to~1 paying~$\costmig \geq \lambda$.
Therefore in any case and by distinguishing $\costmig \leq 2\lambda$ and $\costmig > 2\lambda$,
 we have
\begin{align}	\label{UB:P_firstShort}
  \frac{\ALG(\R_\first)}{\OPT(\R_\first)}
  \leq
  \frac{\mis(\W_{t_\first})+\costmig}{ \min \{2\lambda ,\costmig \}}
  \leq
  \frac{2\lambda+\costmig}{ \min \{2\lambda ,\costmig \}}
  \leq \max\Bigl\{ 4, \frac{\costmig}{\lambda} \Bigr\}
  .
\end{align}

\item
\textbf{$\R_\first$  is long.}
Consider the partitioning
$\R_\first = \U \V \W$,
where $|\U|$ is a~multiple of $3\lambda$,
$|\V| < 3\lambda$ and $|\W|=3\lambda$.
If \OPT serves the entire $\V\W$ in one state (either 0 or 1) then
$\OPT(\V\W) \geq \OPT(\W) \geq \lambda$
(Corollary~\ref{cor:OPT(W)}).
Otherwise, \OPT flips in $\V\W$ and $\OPT(\V\W) \geq \costmig \geq \lambda$.
By applying Observation~\ref{obs:OPT(U)} to $\U$ and using (Observation~\ref{obs:ALG(VW)}),
we obtain
\begin{align}
  \frac{\ALG(\R_\first)}{\OPT(\R_\first)}
  \leq
  \frac{
    \mis(\U) + \mis(\V) + 4\lambda + \costmig
  }{
    \mis(\U) /2 + \lambda
  }
  \leq
  \frac{4\lambda+\costmig}{\lambda} = 4 + \frac{\costmig}{\lambda}
  .
  \label{UB:P_firstLong_noflip}
\end{align}
\end{itemize}
Finally,
we bound costs for the ending segment $\S_\last$ which starts after the last flip to~1 and ends with the last request in~$\seq$.
If $|\S_\last| < 3\lambda$ then
\ALG may perform a~last flip to~0 in this segment and pay up to $3\lambda$ mismatches.
That is,
$\ALG(\S_\last) \leq 3\lambda + \costmig \leq 4\costmig$.
Else,~$|\S_\last | \geq 3\lambda$ and we partition it as
	$\S_\last = \U' \V'$, where~$|\U'|$ is a~multiple of $3\lambda$ and $|\V'|<3\lambda$.
\ALG possibly flips to~0 one last time in a~subsegment $\S'$ of $\S_\last$.
Since $\mis(\V') < 2\lambda \leq 2\costmig$,
and by Observation~\ref{obs:OPT(U)}, we obtain
\begin{align}	\label{UB:P_last}
\ALG(\A_\last) = \ALG(\U' \setminus \S') + + \mis(\S') + \mis(\V') + \costmig
\leq 2\OPT(\U') +  6\costmig
.
\end{align}

From (\ref{UB:P_firstShort}), (\ref{UB:P_firstLong_noflip}), (\ref{UB:P_last}),
and by applying Lemma~\ref{lemma:phaseRatio} to each phase $\P_i$,
we conclude
$\ALG(\seq) \leq (4 + \frac{2\costmig}{\lambda})\OPT(\seq) + 6\costmig$,
where the additive is a~consequence of (\ref{UB:P_last}).
\end{proof}

\noindent
To obtain the proof of Theorem~\ref{thm:6-competitive}, we apply Theorem~\ref{thm:ALG} to all phases of the input sequence.

The challenges in designing time-local algorithms arise from two sources: (1) the algorithm can make decisions only based on the most recent input history, and (2) the algorithm is unaware of its current configuration.
Note that the latter challenge is not present in \emph{memoryless} online algorithms~\cite{Chrobak1991}.
To tackle these challenges, we highlight a useful technique of \emph{tracking} distinguished subsequences of the input as they recede in the visible horizon further toward the past.
Implementing consistent tracking is straightforward in the clocked setting;
we study an example where tracking requests issued at certain points in time is sufficient.

Tracking is significantly more challenging to implement in the non-clocked setting.
Without knowing the temporal position of requests,
requests originating from the same node are indistinguishable.
We overcome this limitation by tracking distinguishable subsequences of requests instead
of single requests.

To conclude, we summarize the analytical and synthetical results for online file migration in Figure~\ref{fig:result-overview}.

\begin{figure}[t]
	\centering
	\includegraphics[width=0.8\textwidth]{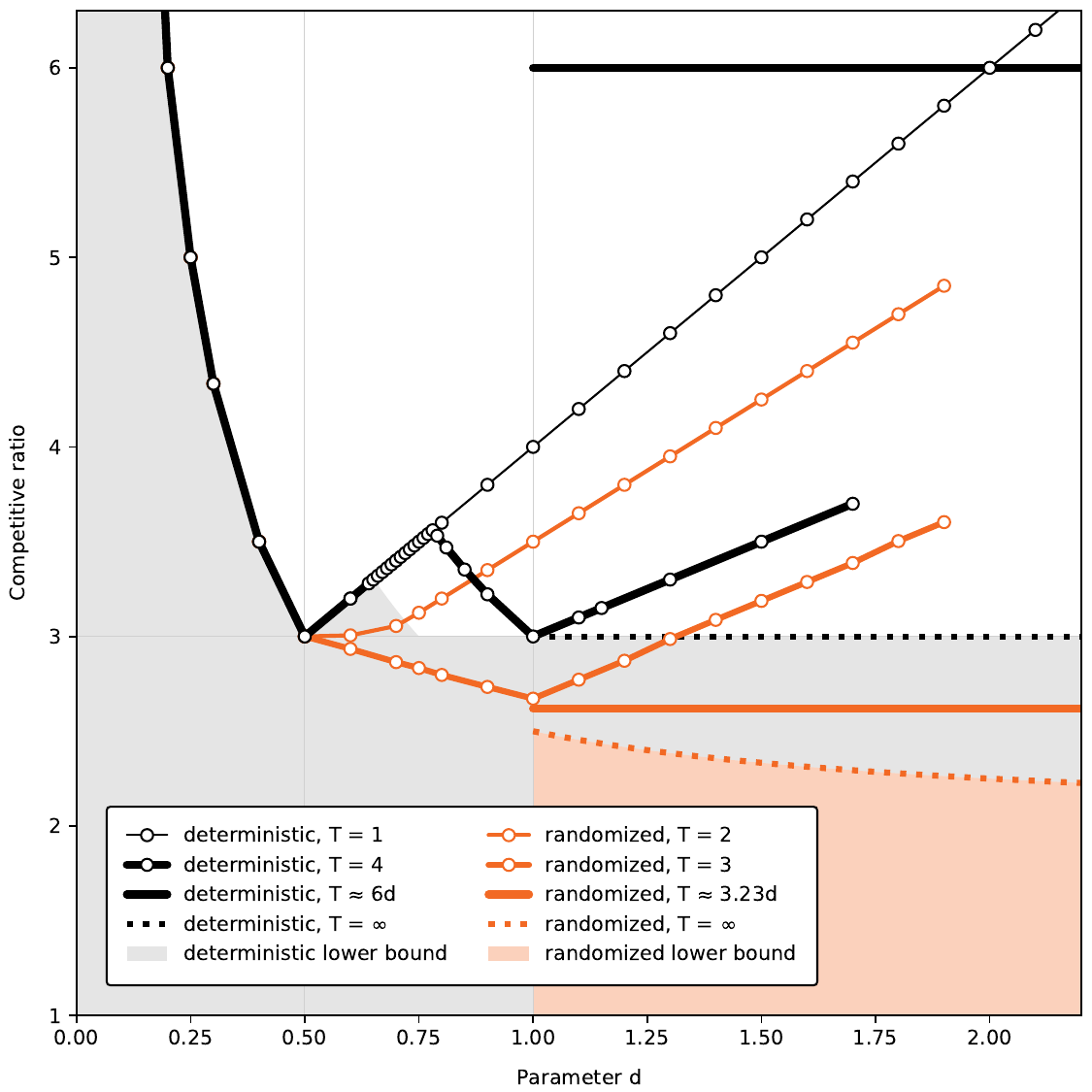}
	\caption{Upper and lower bounds for the online file migration problem. The visualization includes the upper bounds from synthesis (Table~\ref{table:synthesis-results}) for small values of $T$, as well as the upper and lower bounds from the analytical case study (Section~\ref{sec:file-migration}).}
	\label{fig:result-overview}
\end{figure}

\section{Conclusions}
\label{sec:conclusions}

In this work, we initiated the systematic study of time-local online algorithms.
The power of these algorithms comes from the imposed restrictions their decision-making horizon, enabling e.g. automated synthesis, fault recovery, and simplicity of behavior.
Despite their fundamental limitations, we saw that time-local algorithms can solve non-trivial online problems competitively.
We derived the synthesis method for time-local online algorithms, and applied it to the file migration problem, for which problem we followed with analytical and empirical studies.
Furthermore, we defined more general time-local algorithms that include distributed algorithms on paths, and discussed transferability of results between online and distributed algorithms. 

\newpage
\bibliographystyle{plainurl}
\bibliography{time-local}

\vfill

\pagebreak
\appendix

\end{document}